%% file: main.tex
\documentclass[letterpaper,11pt]{article}
\input{header}

\title{Instance-Optimality of Bidirectional Dijkstra on Simple Graphs\footnote{This work was supported by the VILLUM Foundation grant 54451.}}

\author[1]{Christian Bertram}
\author[1]{Mads Vestergaard Jensen}
\author[2]{Mikkel Thorup}
\author[3]{
\and \vspace{-1em} Hanzhi Wang}
\author[1]{Shuyi Yan}

\vspace{0.5em}
\affil[1,2]{BARC, University of Copenhagen \vspace{0.2em}}
\affil[3]{The University of Melbourne \vspace{1.3em}}
\affil[1]{\texttt{\{chbe, mvje, shya\}@di.ku.dk}\vspace{0.2em}}
\affil[2,3]{\texttt{\{mikkel2thorup, hanzhi.hzwang\}@gmail.com}}

\date{}

\begin{document}

    \maketitle

    \input{abstract}
    \newpage
    \setcounter{tocdepth}{2}
    \tableofcontents
    \newpage
    \input{introduction}
    \input{preliminaries}

    \input{technicaloverview}
    \input{undirected_oblivious}

    \input{analysis_directed}
    \input{no_degree}

    \bibliographystyle{alphaurl}
    \bibliography{paper}

\end{document}

%% file: header.tex
\usepackage{amsmath,amssymb,amsfonts, amsthm}
\usepackage{xcolor}
\usepackage{xspace}
\usepackage{thm-restate}
\usepackage{mathtools}
\usepackage{authblk}

\usepackage[numbers, sort]{natbib}

\usepackage{enumerate}
\usepackage{enumitem}
\usepackage{subfiles}
\usepackage{bm}
\usepackage{tabularx}
\usepackage{multirow}
\usepackage{lineno}
\usepackage{stackengine}
\usepackage[T1]{fontenc}
\usepackage[margin=1in]{geometry}
\usepackage[ruled,vlined,linesnumbered]{algorithm2e}
\usepackage[caption=false,font=normalsize,labelfont=sf,textfont=sf]{subfig}
\usepackage{graphicx} 
\usepackage{booktabs}
\usepackage{multirow}
\usepackage{makecell}
\usepackage{amsmath}
\usepackage[normalem]{ulem}
\usepackage{xcolor}
\usepackage{soul}
\SetCommentSty{\color{gray}}
\newcommand{\algcomment}[1]{%
  \hspace*{\fill}%
  \makebox[8.65cm][l]{\textcolor{gray}{// #1}}%
}

\setstcolor{blue}
\setul{0.6ex}{0.25ex}

\setlist[itemize]{leftmargin=*}
\setlist[enumerate]{leftmargin=*}

\usepackage{hyperref}
\hypersetup{
  linktoc      = all, 
  colorlinks   = true, % use colors instead of ugly boxes
  urlcolor     = blue, % color for external hyperlinks
  linkcolor    = blue, % color for internal links
  citecolor    = red, % color for citations
  linktocpage  = true
}
\usepackage[capitalize,noabbrev]{cleveref}
\usepackage{soul}

\newtheorem{theorem}{Theorem}[section]
\newtheorem{lemma}[theorem]{Lemma}

\newtheorem{claim}[theorem]{Claim}

\crefname{claim}{Claim}{Claims}
\Crefname{claim}{Claim}{Claims}

\def\deg{\textup{\textsc{deg}}\xspace}

\def\indeg{\textup{\textsc{indeg}}\xspace}
\def\innbr{\textup{\textsc{in}}\xspace}
\def\outdeg{\textup{\textsc{outdeg}}\xspace}
\def\outnbr{\textup{\textsc{out}}\xspace}

\def\dout{d_{\mathrm{out}}}

\def\setC{\mathcal{C}}

\def\edelta{\hat{\delta}}

\def\cost{cost}
\def\ratio{\phi}

\def\factor{\phi}

\def\eq{E_Q}
\def\hD{h_D}
\def\HD{H_D}

\renewcommand{\P}[1]{{\Pr\!\left[#1\right]}}
\newcommand{\ceil}[1]{{\left\lceil#1\right\rceil}}

\newcommand*{\eps}{\varepsilon}

%% file: abstract.tex
\begin{abstract}

We study the shortest-path problem on graphs with positive real-valued edge weights. Given a source vertex $s$ and a target vertex $t$, the goal is to calculate the length of the shortest path from $s$ to $t$.
We are particularly interested in instances that can be solved in sublinear time.

Recently, Haeupler, Hladík, Rozhoň, Tarjan, and Tětek~\cite{sosa_bidirectional_HaeuplerHRTT25} proved that (a version of) the   bidirectional Dijkstra’s algorithm is \emph{instance-optimal} on positively weighted \emph{multigraphs}, both directed and undirected, considering the number of vertices and edges queried by the algorithm. 
Here, being instance-optimal means best
within a constant factor
on every possible instance, but this is only compared against
algorithms that are correct on every valid instance.

However, multigraphs are not the canonical setting for the shortest-path problem. The problem is typically formulated on simple graphs without loops and parallel edges. \cite{sosa_bidirectional_HaeuplerHRTT25} therefore left as an open problem whether bidirectional Dijkstra remains instance-optimal on simple weighted graphs.

We answer this question, but for simple graphs, the answer is more complex, depending on the setting. We show that bidirectional Dijkstra is still instance-optimal on simple undirected weighted graphs under the order-oblivious model, where incident edges are given in a random order. In contrast, under the order-dependent model, where incident edges have a given order, we show that bidirectional Dijkstra is \emph{not} instance-optimal: it is off by a factor of $\Theta(\min\{\lceil m/n \rceil, n^{1/3}\})$, where $m$ and $n$ denote the number of edges and vertices, respectively. Nonetheless, we show that this is the best possible in the sense that no algorithm can have an instance-optimality ratio of $o(\min\{\lceil m/n \rceil, n^{1/3}\})$.

For simple directed weighted graphs, we show that bidirectional Dijkstra is \emph{not} instance-optimal under either the order-oblivious or the order-dependent model, being off by a factor of $\Theta(\lceil m/n \rceil)$ in both cases. We further show that no algorithm can have instance-optimality ratio $o(\lceil m/n \rceil)$  under the order-dependent model, or under the order-oblivious model when $m=O(n\sqrt{n})$.
On the positive side, the above results imply that bidirectional Dijkstra is instance-optimal up to logarithmic factors on all sparse directed and undirected graphs satisfying $m/n=\log^{O(1)} n$.

Besides the above main bounds, we also study the instance-optimality if we
have an upper bound $\Delta$ on the maximal degree of the graph, and for instances that can be solved in  $o(n)$ time.

Finally we show that most of the instance-non-optimality is due to degree-queries: if there are no degree-queries then bidirectional Dijkstra is order-dependent instance-optimal on undirected graphs and order-oblivious instance-optimal on directed graphs. 
\end{abstract}

%% file: introduction.tex
\section{Introduction} \label{sec:introduction}

The shortest path problem is one of the most fundamental problems in graph algorithms. Given a graph $G$ with positive real-valued edge weights, a distinguished source vertex $s$, and a target vertex $t$,
the problem asks for the length of the shortest path from $s$ to $t$. Often we will also find the shortest paths, but this is not required for lower-bounds.
Let $n$ and $m$ denote the number of vertices and edges in $G$, respectively. The textbook Dijkstra’s algorithm~\cite{Dijkstra59} solves this problem in $O(m+n\log n)$ time using a Fibonacci heap~\cite{FredmanT87_fibonacci} or a relaxed heap~\cite{DriscollGST88_relaxedheap}. This running time is already near-optimal in the worst case, where $\Omega(m)$ time is necessary.

However, when going beyond worst-case analysis and considering graph instances that can be solved  in $o(m)$ time, a number of algorithms have been observed to significantly outperform Dijkstra’s algorithm. One canonical example is bidirectional search~\cite{Dantzig1963, Nicholson66}, which simultaneously runs Dijkstra’s algorithm forward from $s$ and backward from $t$, terminating when the two searches meet. 
Although bidirectional Dijkstra still requires $\tilde \Theta(m)$ time in the worst case, it is often highly efficient in practice, exploring only a small portion of the graph before terminating. 

\paragraph{Instance-optimality for positively weighted multigraphs.} 
A recent work by Haeupler, Hladík, Rozhoň, Tarjan, and Tětek~\cite{sosa_bidirectional_HaeuplerHRTT25} provided a theoretical explanation for the superiority of bidirectional Dijkstra. They showed that, when alternating edge relaxations in the forward and backward searches, this simple algorithm is, surprisingly, \emph{instance-optimal} on all \emph{weighted multigraphs} with real positive weights.  We are going to study what happens for the more standard case of simple graphs, but first we contextualize the result on multigraphs from~\cite{sosa_bidirectional_HaeuplerHRTT25}.

Here, an instance-optimal algorithm~\cite{FaginLN03_instance_optimal, instance_optimal_Roughgarden19} is the dream of an algorithm that is (close to) \emph{best} on \emph{every} single input instance. This covers not just worst-case or average-case instances, but also ``easy'' instances that can be solved in constant time. This contrasts algorithms with optimal worst-case or average-case guarantees that may spend worst-case or average-case time also on such easy instances.
Notably, when discussing instance-optimality, we compete only against algorithms that are \emph{correct} on every valid input instance. Such algorithms may be \emph{instance-smart}, in the sense that they are tuned for a particular instance and may even know the correct answer for that instance in advance. However, since the actual input may differ, an instance-smart algorithm must still \emph{verify} that its answer is correct on the given input. Note that instance-smart algorithms subsume algorithms that receive good but not fully trusted advice (e.g., advice from AI~\cite{MitzenmacherV22}), which likewise must verify their answers on the actual input instance. 

We say that an algorithm $A$ is \emph{instance-optimal within a factor $\factor$}, or that $A$ is \emph{$\factor$-instance-optimal},
if for every input instance $I$, $A$ is at most a factor $\factor$ slower than any correct (instance-smart) algorithm on $I$. The smallest such $\phi$ is called the \emph{instance-optimality ratio} of $A$. When we say an algorithm is instance optimal, we typically mean that $\phi$ is constant.
Such instance-optimality is often considered the ultimate notion of optimality~\cite{colt_ChenL16_openproblem,instnce-optimal_AfshaniBC17}. Finding an algorithm, especially a simple one, that is
instance-optimal is therefore extremely appealing. 

Unfortunately, instance-optimal algorithms exist only in rare situations~\cite{instnce-optimal_AfshaniBC17, FOCS_best_HaeuplerHRTT24}. 
Some positive examples are known, e.g., in quantity estimation~\cite{siamcomp_DagumKLR00}, database aggregation~\cite{FaginLN03_instance_optimal}, sequential estimation~\cite{stoc_ValiantV16_instance_optimal, dang_neurips_23, focs_NarayananRTT24_thorup}, computational geometry~\cite{instnce-optimal_AfshaniBC17, BergHRRW25_esa_instanceopt_convexhull}, distribution testing and learning~\cite{siamcomp/ValiantV17, books_ValiantV20, focs_BlancCW25_clement}, PageRank estimation~\cite{PageRank_TW}, and contextual bandits~\cite{nips_LiRNJJ22}.
As mentioned above,~\cite{sosa_bidirectional_HaeuplerHRTT25} adds bidirectional Dijkstra to this exclusive list, showing that it is instance-optimal on all multigraphs, directed and undirected, with real positive weights.
The instance-optimality ratio is constant when we measure query complexity (number of vertices and edges queried), and at most $O(\log n)$ if we measure the actual running time.

For the above instance-optimality, it is critical that we only allow positive weights, that is, no zero weights. Bidirectional Dijkstra works correctly for the more general case of non-negative weights, but it is no longer instance-optimal. As proved in 
~\cite{sosa_bidirectional_HaeuplerHRTT25}, for the shortest path problem with non-negative weights, there is no algorithm with instance-optimality ratio $o(n/\log n)$. The basic issue is that an instance-smart algorithm may know a ``secret'' path of length  $0$ from $s$ to $t$. Since no shorter path can exist, it just
has to verify this path. In fact, the bad instances are just instances of reachability where all edge weights are zero, which is hence one of the many problems that do not admit an instance-optimal algorithm.

\paragraph{Simple graphs versus multigraphs.}
The instance-optimality from \cite{sosa_bidirectional_HaeuplerHRTT25} is established only for positively weighted multigraphs, allowing both self-loops and parallel edges. Allowing parallel weighted edges may seem redundant for shortest-path computation, but recall that the main obstacle to obtaining instance-optimal algorithms is the existence of instance-smart algorithms tailored to perform exceptionally well on an individual instance. So to investigate the instance-optimality of Dijkstra's algorithm, which relaxes every incident edge when it visits a vertex, one would like to force any correct algorithm to do more or less the same. Parallel edges achieve exactly this: without them, the one-hop distance from $u$ to $v$ is known as soon as a single $(u, v)$ edge is found and its weight read, whereas with parallel edges one must scan the entire out-adjacency list of $u$ (or the in-adjacency list of $v$) to be sure no lighter edge is hidden, which is always possible since edge weights are positive reals.

However, normally, the shortest-path problem is studied on simple graphs. Multigraphs reduce to this case simply by removing all loops and merging parallel edges into a single edge with the smallest weight available. It is therefore natural to ask what happens on simple graphs.
This question was left open by~\cite{sosa_bidirectional_HaeuplerHRTT25}, whose concluding remarks (where Theorem~1.1 denotes their instance-optimality result for positively weighted multigraphs) state: 
\begin{quote}
``Our proof of Theorem 1.1 works for multigraphs that allow parallel edges and self-loops. While we believe that the possibility of self-loops can be avoided, we are not sure whether the same holds for parallel edges: We believe that whether Theorem 1.1 holds in the setting of simple graphs is an interesting open question.'' ~\cite{sosa_bidirectional_HaeuplerHRTT25}
\end{quote}

\begin{table}[!t]
\centering
\small
\setlength{\tabcolsep}{6pt}
\renewcommand{\arraystretch}{1.4}
\begin{tabularx}{\linewidth}{@{}l l >{\raggedright\arraybackslash}X >{\raggedright\arraybackslash}X@{}}
\toprule
\makecell[l]{\textbf{Graph with} \\ \textbf{degree query}} & \makecell[l]{{\bf Order}\\{\bf model}}
 & \makecell[l]{ {\bf Instance-optimality ratio}\\ {\bf of bidirectional Dijkstra}}
 & \makecell[l]{ {\bf Instance-optimality ratio} \\ {\bf of every good algorithm}} \\
\midrule
\multirow{4}{1.9cm}{Undirected}
    & oblivious
    & $O(1)$~[Thm~\ref{thm:undirected-oblivious-undirected}] (\textbf{instance-optimal})
    & $\Omega(1)$ (trivial) \\
\cmidrule(lr){2-4}
    & \multirow{2}{*}{\makecell[l]{dependent}}
    & $O(\min\{\lceil m/n \rceil,\,n^{1/3}\})$~[Thm~\ref{thm:undirected-fixed-m}]
       & $\Omega(\lceil m/n \rceil)$ if $m=O(n^{4/3})$~[Thm~\ref{thm:undirected-fixed-lower-m}]  \\
  & & $O(\lceil\Delta^{2/3}/n^{1/3}\rceil)$~[Thm~\ref{thm:undirected-fixed}]
    & $\Omega(\lceil\Delta^{2/3}/n^{1/3}\rceil)$~[Thm~\ref{thm:undirected-fixed-lower}]  \\
    & & $O(1)$ if $\tau\le n/20$ [Thm~\ref{thm:undirected-tau}] &
    \\
\midrule
\multirow{5}{1.9cm}{Directed}
& \multirow{2}{*}{\makecell[l]{oblivious}}
    & $\Omega(\lceil m/n \rceil)$ [Thm~\ref{thm:directed-oblivious-m}]
    & $\Omega(\lceil m/n \rceil)$ if $m=O(n\sqrt{n})$~[Thm~\ref{thm:directed-oblivious-m-any}] \\
  & &
   $\Omega(\lceil\Delta^2/n\rceil)$~[Thm~\ref{thm:directed-oblivious}]   &  \\
\cmidrule(lr){2-4}
& \multirow{2}{*}{\makecell[l]{dependent}}
  & $O(\lceil m/n \rceil)$~[Thm~\ref{thm:directed-fixed-d}]
    & $\Omega(\lceil m/n \rceil)$~[Thm~\ref{thm:directed-fixed-lower-m}] \\
 & & $O(\lceil\Delta^2/n\rceil)$ [Thm~\ref{thm:directed-fixed}]
    & $\Omega(\lceil\Delta^2/n\rceil)$ [Thm~\ref{thm:directed-fixed-lower}] \\
    & & $O(h)$ if $\tau\le n/20$ [Thm~\ref{thm:directed-tau}] &
    \\
\bottomrule
\end{tabularx}
\caption{Instance-optimality ratios for shortest path on simple weighted graphs with positive real weights. Note that any upper bound on the instance-optimality ratio established under the order-dependent model extends to the order-oblivious model.
We thus have tight bounds on bidirectional Dijkstra in almost every setting.
Here, $n$ and $m$ denote the numbers of vertices and edges in $G$, respectively,  and $\Delta$ denotes the maximum in- and out-degree of $G$. 
Additionally, $\tau=\tau(I)$ denotes the instance-wise complexity, i.e., the minimum of the expected number of queries made by all good algorithms on instance $I$. Thus, the condition $\tau<n/10$ means that the instance-optimality result holds for all instances whose instance-wise complexity satisfies $\tau<n/10$.
$h\ge1$ is an upper bound of the number of vertices with degree at least $n-10\tau$.
}
\label{tab:results}
\end{table}

\begin{table}[!t]
\centering
\small
\setlength{\tabcolsep}{6pt}
\renewcommand{\arraystretch}{1.4}
\begin{tabularx}{\linewidth}{@{}l l >{\raggedright\arraybackslash}X >{\raggedright\arraybackslash}X@{}}
\toprule
\makecell[l]{\textbf{Graph w/o} \\ \textbf{degree query}} & \makecell[l]{{\bf Order}\\{\bf model}}
 & \makecell[l]{ {\bf Instance-optimality ratio}\\ {\bf of bidirectional Dijkstra}}
 & \makecell[l]{ {\bf Instance-optimality ratio} \\ {\bf of every good algorithm}} \\
\midrule
\multirow{2}{1.9cm}{Undirected}
  & oblivious
    & $O(1)$~[Thm~\ref{thm:order-dependent-undirected-nodegree}] (\textbf{instance-optimal})
    & $\Omega(1)$ (trivial) \\
\cmidrule(lr){2-4}
& dependent  & $O(1)$~[Thm~\ref{thm:order-dependent-undirected-nodegree}] (\textbf{instance-optimal}) & $\Omega(1)$ (trivial)  \\
\midrule
\multirow{2}{1.9cm}{Directed}
& oblivious  & $O(1)$~[Thm~\ref{thm:oblivious-directed-nodegree}] (\textbf{instance-optimal})
    & $\Omega(1)$ (trivial)  \\
\cmidrule(lr){2-4}
& dependent & $O(\lceil\Delta^2/n\rceil)$ [Thm~\ref{thm:directed-fixed}] & $\Omega(\lceil\Delta/\sqrt n\rceil)$ [Thm \ref{thm:directed-dependent-no-degree-lower}]\\
\bottomrule
\end{tabularx}
\caption{Instance-optimality ratios without degree queries for shortest path on simple weighted graphs with positive real weights. All symbols have the same explanations as in Table~\ref{tab:results}. 
}
\label{tab:results_nodegree}
\end{table}

\paragraph{Our contribution.}
We answer the above question, but the answer is more complex for positively weighted \emph{simple graphs}, depending on the setting (the result from \cite{sosa_bidirectional_HaeuplerHRTT25} was that bidirectional Dijkstra is instance-optimal for positively weighted \emph{multigraphs} in all the settings mentioned below). We show that bidirectional Dijkstra is still instance-optimal on simple undirected weighted graphs under the order-oblivious model, where incident edges are given in a random order. In contrast, under the order-dependent model, where incident edges have a given order, we show that bidirectional Dijkstra is \emph{not} instance-optimal: it is off by a factor of $\Theta(\min\{\lceil m/n \rceil, n^{1/3}\})$. Nonetheless, we show that this is best possible, in the sense that every algorithm has instance-optimality ratio $\Omega(\min\{\lceil m/n \rceil, n^{1/3}\})$

For simple directed weighted graphs, we show that bidirectional Dijkstra is \emph{not} instance-optimal under either the order-oblivious or the order-dependent model, being off by a factor of $O(\lceil m/n \rceil)$ in both cases. We further show that every algorithm has an $\Omega(\lceil m/n \rceil)$ instance-optimality ratio under the order-dependent model, and under the order-oblivious model when $m=O(n\sqrt{n})$.
On the positive side, the above results imply that bidirectional Dijkstra is instance-optimal up to logarithmic factors on all sparse directed and undirected graphs satisfying $m/n=\log^{O(1)} n$.

Besides the above main bounds, we also study the instance optimality ratio if we
have an upper bound $\Delta$ on the maximal degree, and for instances that can be solved in $o(n)$ time. A summary of our results is given in~\Cref{tab:results}. We will present these results in~\Cref{sec:more-results} after we have given a more detailed description of the model in~\Cref{sec:model}.

Finally we show that most of the instance-non-optimality is due to degree-queries: bidirectional Dijkstra's algorithm does not use degree-queries, and if degree-queries are not supported, then bidirectional Dijkstra is instance-optimal on both directed and undirected positively weighted graphs in the order-oblivious model. The only exception to instance-optimality is for directed graphs in the order-dependent model. 
The results without degree queries are given in Table \ref{tab:results_nodegree}.

\paragraph{Relation to A$^*$ heuristic search and advice from AI.}

As discussed in  \cite{sosa_bidirectional_HaeuplerHRTT25}, the popular A$^*$ search algorithm~\cite{hart1968formal} is a heuristic variant of Dijkstra’s algorithm that is sometimes observed to run faster than bidirectional Dijkstra in practice. The A$^*$ algorithm relies on a heuristic estimate $\edelta(v)$ of the distance from each vertex $v$ to the destination $t$, and runs Dijkstra’s algorithm on the adjusted edge weights $\ell'(u,v)=\ell(u,v)+\edelta(v)-\edelta(u)$, where $\ell(u,v)$ denotes the original weight of edge $(u,v)$. If $\edelta$ is exact and $\edelta(v)=d(v,t)$ for all $v$, then all edges on these shortest paths have adjusted weight $0$, and A$^*$ will only explore vertices on shortest paths from $s$ to $t$. 

Since Dijkstra does not work with negative edge-weights and A$^*$ builds on Dijkstra, the correctness of the A$^*$ heuristic search requires that the adjusted edge weight $\ell'(u,v)$ is non-negative for every edge $(u,v)$; otherwise A$^*$ may return an incorrect shortest $s$-$t$ path. The heuristic estimate $\edelta$ is therefore defined to be \emph{consistent} only if all adjusted weights are non-negative, that is, if $\ell(u,v)+\edelta(v)\ge \edelta(u)$ holds for every edge $(u,v)$ in the graph.
Verifying consistency of $\edelta$ requires checking the adjusted weights of all edges and therefore takes $\Omega(m)$ time.
This nearly matches the $\tilde{O}(m)$ worst-case running time of Dijkstra, so in practice, with A$^*$, we just have to trust that $\edelta$ is consistent. With this trust in $\edelta$, the A$^*$ search may outperform bidirectional Dijkstra, but this does not contradict instance-optimality of bidirectional Dijkstra since guaranteed correctness assumes the extra information that $\edelta$ is consistent with respect to the given input graph. This point with a trusted $\edelta$ was already discussed in \cite{sosa_bidirectional_HaeuplerHRTT25}.

Suppose now that we run A$^*$ with a $\edelta$ that we do not fully trust to be consistent on the given input graph, e.g., $\edelta$ could be
an ``advice'' on the input provided by an AI system that we do not fully trust~\cite{MitzenmacherV22}. To be sure of correctness, we must verify that the path found is indeed a shortest path in the input graph, just as any other instance-smart algorithm must. Consequently, we cannot be significantly faster than an instance-optimal bidirectional Dijkstra that constructs the shortest path from scratch. Here, by an instance-optimal bidirectional Dijkstra we mean one of the cases where bidirectional Dijkstra is instance-optimal, e.g., positively weighted multi-graphs as proved in \cite{sosa_bidirectional_HaeuplerHRTT25}, or undirected simple graphs in the order-oblivious model as proved in this paper.

\subsection{Model}\label{sec:model}
\def\cI{\mathcal I}
\def\cC{\mathcal C}

\paragraph{Good algorithms.}
Recall that when we discussed instance-optimality, we had a set $\cI$ of \emph{valid instances}, and  we compared only with correct algorithms, namely algorithms that return a \emph{correct} answer on every $I\in\cI$. However, we can also compare with Monte Carlo randomized algorithms whose correctness is probabilistic. Accordingly, we say that a randomized algorithm $A$ is \emph{good} if on every valid instance it returns a correct answer with probability at least $9/10$. The choice of the constant $9/10$ is arbitrary, since standard probability amplification can boost any constant success probability greater than $1/2$ to $9/10$.

\paragraph{Instance-wise complexity and instance-optimality ratio.} 

Recall that we have introduced the instance-optimality ratio $\ratio$ of an algorithm $A$ as the smallest factor by which $A$ is slower than any correct algorithm on every valid input instance $I\in\cI$. We now define this notion more formally. 
  
Let $\cost_A(I)$ denote the expected number of queries made by an algorithm $A$ on an instance $I$. The \emph{instance complexity} $\tau=\tau(I)$ of $I$ is then defined as the minimum $\cost_A(I)$ for any good algorithm $A$ on the instance $I$. Note that for the instance complexity of $I$, we do not care how much time $A$ spends on other instances, so $A$ could be instance-smart, specially tuned for $I$. However, to be good, $A$  has to verify the answer on $I$ so that it does not provide a wrong answer because it gets a different input $I'$.
In particular, we note that $\tau(I)$ cannot decrease if we increase the set $\cI$ of valid instances that algorithms have to be good on.

The \emph{instance-optimality ratio} of $A$ is defined as 
\begin{align}\label{eqn:ratio}
\ratio_A = \sup_{I\in \cI}\{\cost_A(I)/\tau(I)\},
\end{align}
that is, the smallest value of $\ratio_A$ such that $\ratio_A\ge \cost_A(I)/\tau(I)$ for every valid instance $I$.
We call $A$ \emph{instance-optimal} if its instance-optimality ratio $\ratio_A$ is a constant. When $\phi_A$ is not constant, we may want to parameterize it by an instance-parameter such as the size $n$, e.g., 
\begin{align}\label{eqn:ratio-par}
\ratio_A(n) = \sup_{I\in \cI, |I|=n}\{\cost_A(I)/\tau(I)\}. 
\end{align}
Finally, it may be relevant to study the instance-optimality ratio of $A$ on a subset $\mathcal{C}$ of valid instances as
\begin{align*}
\ratio_A(\setC) = \sup_{I\in \setC}\{\cost_A(I)/\tau(I)\}. 
\end{align*}
We call $A$ \emph{instance-optimal} on $\setC$ if $\ratio_A(C)$ is a constant.

An example from \cite{sosa_bidirectional_HaeuplerHRTT25} is that if
the valid instances are multigraphs with positive integer weights, then the instance-optimality ratio of bidirectional Dijkstra is $\Theta(\Delta)$ where $\Delta$ bounds the maximal degree. However, if the valid instances include weight zero, then the instance-optimality ratio is constant on instances with only positive weighs. Another example from \cite{PageRank_TW} is that a certain PageRank algorithm is instance-optimal except on digraphs with many vertices of very high degree. The latter example illustrates that even if an algorithm is not instance optimal on all the valid input instances (all digraphs), we can still hope that is instance-optimal on a rich and interesting subclass.

\paragraph{Query model for sublinear graph algorithms.}
We are particularly interested in instances on which the shortest-path problem can be solved in $o(m)$ time. Like in \cite{sosa_bidirectional_HaeuplerHRTT25}, we therefore work in the standard \emph{adjacency-list} query model for sublinear graph algorithms. 

We assume that the input is a simple graph, with no self-loops or parallel edges; the graph has $n$ vertices, and all edges have positive real-valued weights.
Additionally, we assume the number of vertices $n$ of the graph is sufficiently large and is known in advance. 
The graph is stored in incidence-list format. The vertices are numbered $1$ to $n$, and each is identified by its index; an algorithm has query access to all of them. Each vertex keeps a list of its outgoing edges and a list of its incoming edges. Given a vertex $v\in[1,n]$, the model supports \emph{degree queries} $\indeg(v)$ and $\outdeg(v)$, which return the in- and out-degree of $v$, and \emph{neighbor queries} $\innbr(v,i)$ and $\outnbr(v,i)$, which return the $i$-th in- and out-neighbor of $v$ together with the weight of the corresponding edge (if $i$ exceeds the length of the relevant list, the query returns null). Each query has unit cost. We measure the complexity of an algorithm by its query complexity, i.e., the total number of queries it performs.

Throughout the paper we consider both directed and undirected graphs, but we state our model and algorithms only for the directed case: they naturally transfer to undirected graphs, since an undirected graph is simply a directed graph in which each vertex has the same in- and out-degree and the same in- and out-adjacency lists.

\paragraph{Order-dependent and order-oblivious.}
The shortest path problem normally refers to a graph $G=(V,E)$ with a set $V$ of vertices and a set $E$ of edges. However, when we store the graph for the above query model, we have to order vertices, indexing them $1,\ldots,n$. Moreover, for each vertex $v$, we have to order the incident edges so that we can index the neighbors as $\innbr(v,i)$ and $\outnbr(v,i)$,

When we discuss instance-optimality, it is important whether the orderings are considered part of the input or whether they are just random and provide no information beyond the graph itself. The former case is called \emph{order-dependent} while the latter is \emph{order-oblivious}. 

For an instance-smart algorithm, it is better to be order-dependent, since it offers more information on the instance that it is tuned for. 
For example, the instance-smart algorithm may know the edge indices on the shortest path from $s$ to $t$. All it has to do is to verify this path, and verify that there is no shorter path. In the case of zero-weight edges, it only has to verify the zero-length path, without having to verify that no shorter path exists in the graph. 

From the perspective of the algorithm that hopes to be instance-optimal, the situation is the opposite. Even if the order is part of the instance, we still need to be best for all possible orderings, so we can think of the order as adversarially chosen. Internally, using randomization, it can always simulate a random order, and for bidirectional Dijkstra, it makes no asymptotic difference whether the order is random or adversarial, so we will always think of the order as adversarial. 

Since the order is not a natural part of the graph instance, it may be fairer to say that, as in the order-oblivious setting, the order is not known to the instance-smart algorithm. When it starts, it only knows the indices of $s$ and $t$, say, 1 and $n$. 
For example, in applications where vertices do not have comparable identifiers (e.g., webpages on the Internet, which are identified only by URLs), an adjacency list is more naturally viewed as an unordered set, and a neighbor query is usually implemented by returning a previously unseen neighbor chosen uniformly at random. 
Note that in the order-oblivious model, even though we think of the ordering as random,  we still require a good algorithm to be correct with probability at least $9/10$ for any possible ordering of any input instance.

We also emphasize that the order-oblivious model applies to the ordering of the vertex list. In particular, vertex indices cannot be considered part of the input, and accessing an unseen vertex by its index in the order-oblivious model should be viewed as sampling a vertex uniformly at random from the graph.''

It follows that any upper bound on the instance-optimality ratio established under the order-dependent model also holds under the order-oblivious model. Dually, any lower bound established under the order-oblivious model extends to the order-dependent model.

\subsection{Additional results}\label{sec:more-results}

Recall that \cite{sosa_bidirectional_HaeuplerHRTT25} proved that bidirectional Dijkstra is instance-optimal for positively weighted \emph{multigraphs}, both directed and undirected, in the order-dependent model, and hence also in the order-oblivious model. They left open the more typical case of \emph{simple graphs}. 

As discussed, we answer this question by showing that it remains instance-optimal for undirected graphs in the order-oblivious model, but not in any of the other models. In terms of $n$ and $m$, the instance-optimality ratio is $\Theta(\min\{\lceil m/n \rceil, n^{1/3}\})$ for undirected graphs in the order-dependent model and $\Theta(\lceil m/n \rceil)$ for directed graphs in both the order-dependent and order-oblivious models. These results imply that bidirectional Dijkstra is instance-optimal for $m=O(n)$.

In the following, we present several additional results that further characterize the instance-optimality ratio. All the results are summarized in~\Cref{tab:results}, and all of them are understood to be for positively weighted simple graphs, directed or undirected. 

\paragraph{Max-degree.}
If we have a bound $\Delta$ on the maximal degree, then for directed graphs, we show that the instance-optimality ratio of bidirectional Dijkstra is $\Theta(\lceil \Delta^2/n \rceil)$ in both the order-dependent and order-oblivious models. For undirected graphs, the bound in the order-dependent model is the cube root of this quantity, that is, $\Theta(\lceil \Delta^{2/3}/n^{1/3} \rceil)$. In all cases, bidirectional Dijkstra is instance-optimal when $\Delta=O(\sqrt{n})$. 
More generally, we show that in the order-dependent model, the instance-optimality ratio of bidirectional Dijkstra is $O(\hD+\lceil D^2/n\rceil)$ on directed graphs with at most $h$ vertices whose in-degree or out-degree is larger than $D$.

\paragraph{Instance-optimality on easy instances.}
Recall that our whole motivation for considering bidirectional Dijkstra is that it sometimes only has to consider a small fraction of the graph.
We are therefore particularly interested in ``easy'' cases with instance-complexity $\tau(G)=o(n)$, that is, where some instance-smart algorithm can verify the shortest path in $o(n)$ time. We claim that this setting is also very beneficial for bidirectional Dijkstra.

First, considering the opposite case where the instance-complexity is $\Omega(n)$, then the instance-optimality ratio of $O(\ceil{m/n})$ is obtained by any algorithm with worst-case complexity $O(m)$. Our instance-optimality ratio upper-bound of $O(\ceil{m/n})$ is thus only interesting when the instance-complexity is $o(n)$. 
However, for the lower-bound, with instance-complexity $\Omega(n)$, an instance-smart algorithm can record the degree of all $n$ vertices, and this can be very useful information, allowing
it to solve some instances in $O(n)$ time where bidirectional Dijkstra uses $\Theta(m)$ time. 

Considering the case of ``easy'' instances with instance-complexity $o(n)$ we show that bidirectional Dijkstra is instance-optimal for undirected graphs, also in the order-dependent model.
Thus, if an undirected graph is easy for any good algorithm, then it is equally easy for bidirectional Dijkstra. 

For easy instances, we also have improved an instance-optimality ratio for directed graphs. If a graph $G$ has instance-complexity $\tau\leq n/20$ and at most $h\geq 1$ vertices of degree bigger than $n-10\tau$, then the instance-optimality ratio is $O(h)$. 
This bound is much better than our $O(\ceil{\Delta^2/n})$ bound, for if $h\geq 1$, then $\Delta\geq n/2$. We also note that $h\leq 2m/n$, so $O(h)$ is no worse than $O(\lceil m/n \rceil)$, but typically much better because there are often only few vertices of degree close to $n$ in real-world graphs~\cite[Section 4.5]{networkscience}, e.g., $h=O(\log n)$ on all scale-free networks. Consequently, bidirectional Dijkstra achieves instance-optimality up to polylogarithmic factors on such networks if they admit an $o(n)$ solution.

\paragraph{Valid weights.}
Like in \cite{sosa_bidirectional_HaeuplerHRTT25}, our main focus in this paper is on graphs with positive real weights. Let $d(s,u)$ denote the distance from $s$ to $u$. 
What we really need is that for any vertices $u$ and $v$, if $d(s,u)+d(v,t)<d(s,t)$, then there should be a valid weight $\eps<\frac{d(s,t)-(d(s,u)+d(v,t))}2$ (which could be used on an alternative length-2 path from $u$ to $v$). This means that our results also apply to instances with positive floating point numbers as long as they have not used the full precision allowed in valid weights. Likewise they hold if we say that valid weights are all non-negative integers, but we only consider the instance-optimality ratio on instances with positive integer weights.

In \cite{sosa_bidirectional_HaeuplerHRTT25} they show that these restrictions are necessary for the instance-optimality of bidirectional Dijkstra.
If we consider instances with a smallest valid positive weight, e.g., if the valid weights are the positive integers, then the instance-optimality ratio is $\Theta(\Delta)$.
Note that for $\Delta=o(n)$, this is worse than our general $O(\ceil{\Delta^2/n})$ bound.
Also, they show a lower bound of $\Omega(n/\log n)$ for the instance-optimality ratio of any algorithm if we consider instances with weights zero.

These negative results from \cite{sosa_bidirectional_HaeuplerHRTT25} are all for multigraphs, but they are quite easily translated to simple graphs.

\paragraph{No degree-query.}
We will also consider what happens if there are no degree-queries. In that case it is most natural to say that edges are in a real linked list where you pay every time you visit the next edge. Thus you will only know the degree when you get to the last edge. This is, in fact, one of the most classic representation of a graph, and bidirectional Dijkstra works for this representation. 
\Cref{tab:results_nodegree} summarizes our results on the instance-optimality ratio when degree queries are unavailable.

Corresponding to the order-dependent model, for any given graph instance $G$, we consider the worst performance of bidirectional Dijkstra (or any other algorithm that we study instance-optimality for) over all possible orderings of the incidence lists. However, for the instance-complexity, we consider the best time over all possible orderings of the incidence list. This includes the ordering an order-dependent algorithm would have chosen, but this way, a good algorithm has to visit the whole list to determine the degree. 
Even in this strong model, we will show that bidirectional Dijkstra is instance-optimal on undirected graphs, but is not instance-optimal on directed graphs.

Corresponding to the order-oblivious model, we will consider the
expected performance where the incidence lists are randomly ordered. We show that, bidirectional Dijkstra is instance-optimal on all positively real weighted directed and undirected graphs. 
Note that this is the strongest model for lower bounds, and here we also show that the negative bounds with other weights still hold for simple graphs.
That is, if we have zero-weights, then no algorithm can have instance-optimality ratio below $\Omega(n/\log n)$, and if we study instances with a minimal positive weight, then the instance optimality ratio is $\Theta(\Delta)$. Here we cannot just use the construction from \cite{sosa_bidirectional_HaeuplerHRTT25}, for in the order-oblivious case, they used degrees to code vertex indices, but the issue is simple to fix.

\subsection{Paper Organization}

We start by revisiting the details and correctness of bidirectional Dijkstra in~\Cref{sec:preliminaries}.
Next, we set up the general framework for our analysis in~\Cref{sec:framework}.

In~\Cref{sec:order-oblivious-undirected}, as our first result, we present our main positive answer to the open problem from~\cite{sosa_bidirectional_HaeuplerHRTT25}; namely that bidirectional Dijkstra is instance-optimal on undirected graphs in the order-oblivious model.

Afterwards, in~\Cref{sec:order-dependent-directed,sec:order-oblivious-directed,sec:order-dependent-undirected}, we present our main negative answers, proving lower bounds on the instance-optimality ratio in the remaining settings.
We present these negative results in the order that we deem easiest to follow technically, covering order-dependent directed graphs in~\Cref{sec:order-dependent-directed}, order-oblivious directed graphs in~\Cref{sec:order-oblivious-directed}, and order-dependent undirected graphs in~\Cref{sec:order-dependent-undirected}.

Furthermore, in~\Cref{sec:more-positive}, we present additional positive results, giving upper bounds on the instance-optimality ratio that in many cases match our lower bounds, as well as strong upper bounds for easy instances.

Then in~\Cref{sec:nodegree}, we present our results when degree queries are unavailable. In~\Cref{sec:zero-unit-weight}, we briefly discuss the cases where zero-weight edges are allowed, or the input graph is promised to have a smallest valid positive edge weight. 

%%% Local Variables:
%%% mode: LaTeX
%%% TeX-master: "main"
%%% End:

%% file: preliminaries.tex
\section{Revisiting the Bidirectional Dijkstra's Algorithm}\label{sec:preliminaries}

Since we will investigate the instance-optimality of bidirectional Dijkstra in this paper, we revisit the algorithm in this section. 

As mentioned earlier, bidirectional Dijkstra was first described in~\cite[Section 17-3.B]{Dantzig1963}, albeit only briefly and without a formal correctness proof. More detailed descriptions of the algorithm, together with several variants and their analyses, were later presented in~\cite{Nicholson66, Pohl69}.

\subsection{Dijkstra's Algorithm}
At a high level, bidirectional Dijkstra runs the standard Dijkstra's algorithm from the source vertex $s$ (the \emph{forward search}) and the target vertex $t$ (the \emph{backward search}) simultaneously. 
We present the pseudocode for the forward search in~\Cref{alg:forward}. The backward search is analogous, except that it traverses incoming edges instead of outgoing edges.

\begin{algorithm}[H]
\DontPrintSemicolon
\small
\caption{Forward$(s)$}
\label{alg:forward}

\KwIn{graph $G(V,E)$, source vertex $s$ \algcomment{backward: target vertex $t$}}

$\hat d(s,\cdot) \gets +\infty;\; \hat d(s,s) \gets 0$
\hfill \textcolor{gray}{// backward: $\hat d(\cdot,t)\gets+\infty,\; \hat d(t,t)\gets0$}\;

Open $s$
\hfill \textcolor{gray}{// backward: open $t$}\;

\While{\textup{an open vertex exists}}{
    $u_s \gets \arg\min \hat d(s,x)$ among all open vertices $x$
    \hfill \textcolor{gray}{// backward: select $u_t$ minimizing $\hat d(x,t)$}\;
    Close $u_s$ and open all unclosed out-neighbors of $u_s$ \quad \textcolor{gray}{// backward: replace out- by in-neighbors}\;
    \For{\textup{each out-neighbor} $x$ \textup{of} $u_s$}{
            $\hat d(s,x) \gets \min\{\hat d(s,x),\, \hat d(s,u_s)+\ell(u_s,x)\}$  \quad \textcolor{gray}{// backward: relax $\hat d(x,t)$ via in-edge $(x,u_t)$}\;
    }
}
\end{algorithm}

To describe the forward search in more detail, Dijkstra’s algorithm maintains a tentative distance $\hat d(s,u)$ from the source vertex $s$ to each vertex $u\in V$. Initially, $\hat d(s,s)=0$, while $\hat d(s,u)=\infty$ for every $u\neq s$. The value $\hat d(s,u)$ can be interpreted as the length of the shortest path from $s$ to $u$ discovered so far and therefore serves as an upper bound on the true shortest-path distance $d(s,u)$.

Moreover, Dijkstra’s algorithm assigns each vertex one of three states: \emph{open}, \emph{closed}, or \emph{unlabeled}. Initially, the source vertex $s$ is open, and every other vertex is unlabeled. The algorithm repeatedly selects the open vertex $u$ with the smallest tentative distance $\hat d(s,u)$ and closes it. For every unlabeled out-neighbor $v$ of $u$, the algorithm changes its state to \emph{open}. It then \emph{relaxes} every outgoing edge $(u, v)$ for such $v$ by updating 
$\hat d(s,v) \gets \min\{\hat d(s,v),\,\hat d(s,u)+\ell(u, v)\}$,     
where $\ell(u,v)$ denote the weight of edge $(u,v)$. Updating $\hat d(s,v)$ implies that a shorter path from $s$ to $v$ through $u$ is found. The basic Dijkstra’s algorithm terminates when no open vertex exists, and regards $\hat d(s,u)$ as the true distance $d(s,u)$ for all vertices $u$. 

By the correctness of Dijkstra’s algorithm, we have the following lemma: 
\begin{lemma}[\cite{Nicholson66}]\label{lem:invariant_dijkstra}
In both the forward and backward searches, Dijkstra's algorithm maintains the following invariants: 
\begin{enumerate}[label=(\alph*), font=\normalfont]
\item Every closed vertex has exact tentative distance. That is, $\hat d(s,u)=d(s,u)$ for every forward-closed vertex $u$, and $\hat d(u,t)=d(u,t)$ for every backward-closed vertex $u$.
\item Vertices are closed in nondecreasing order of their true distances. That is, if $u$ is forward-closed before $v$, then $d(s,u)\le d(s,v)$; if $u$ is backward-closed before $v$, then $d(u,t)\le d(v,t)$. Equivalently, when a vertex $v$ is forward-closed (resp.\ backward-closed), every vertex $u$ satisfying $d(s,u)<d(s,v)$ (resp.\ $d(u,t)<d(v,t)$) has already been forward-closed (resp.\ backward-closed).
\end{enumerate}
\end{lemma}

\subsection{Bidirectional Dijkstra}
As mentioned above, bidirectional Dijkstra performs a forward search from the source vertex $s$ and a backward search from the target vertex $t$ simultaneously. Different variants of bidirectional Dijkstra differ in how they alternate between the forward and backward searches (e.g., after relaxing a single edge or after completing the relaxation of all outgoing edges of a vertex), as well as in their stopping conditions (e.g., terminating when the two searches meet or terminating when a distance threshold is satisfied). We refer the reader to~\cite{sosa_bidirectional_HaeuplerHRTT25, Dreyfus69} for a brief overview of these variants.

In particular, \cite{sosa_bidirectional_HaeuplerHRTT25} proposes a variant that alternates between relaxing one edge in the forward search and one edge in the backward search, and terminates once $\hat d(s,u_s)+\hat d(u_t,t)\ge \mu$, where $\mu$ denotes the length of the shortest $s$–$t$ path found so far. 

Initially, $\mu \gets \infty$. Whenever the algorithm relaxes an edge, it also updates $\mu$ if a shorter $s$–$t$ path is discovered. For example, when relaxing an edge $(u_s,x)$ in the forward search, the algorithm additionally updates 
\begin{align}\label{eqn:update-mu}
\mu \gets \min\{\mu, \; \hat{d}(s, u_s)+\ell(u_s,x)+\hat{d}(x, t)\}.     
\end{align}
This is because $\hat d(s,u_s)+\ell(u_s,x)+\hat d(x,t)$ is the length of the shortest $s$–$t$ path currently known that passes through the edge $(u_s,x)$. 
If this length is shorter than $\mu$, then a shorter $s$-$t$ path has been found, and we therefore update $\mu$ accordingly. 
Note that this length is $\infty$ if $\hat{d}(x,t)=\infty$.
An analogous update is performed whenever an edge is relaxed in the backward search.

When the termination criterion $\hat d(s,u_s)+\hat d(u_t,t)\ge \mu$ is satisfied, the algorithm can certify that no $s$–$t$ path shorter than $\mu$ exists. It therefore terminates and returns $\mu$. We present the pseudocode of the algorithm in~\Cref{alg:bi-dj-new}. Now each vertex maintains an independent state (open, closed, or unlabeled) in each of the forward and backward searches. 
\cite{sosa_bidirectional_HaeuplerHRTT25} proves that the algorithm is instance-optimal up to a constant factor when the class of valid inputs consists of all weighted multigraphs, whether directed or undirected.

\begin{algorithm}[t!]
\DontPrintSemicolon
\small{
\caption{Bidirectional Dijkstra's Algorithm}
\label{alg:bi-dj-new}
\KwIn{graph $G(V,E)$, source vertex $s$, target vertex $t$}
\KwOut{length of the shortest path from $s$ to $t$}
$\hat d(s,\cdot) \gets +\infty;\; \hat d(s,s) \gets 0$ \quad \textcolor{gray}{// initialize forward search from $s$}\;
$\hat d(\cdot, t) \gets +\infty;\; \hat d(t,t) \gets 0$ \quad \textcolor{gray}{// initialize backward search from $t$}\;
Mark $s$ as forward-open and $t$ as backward-open\;
$\mu \gets \hat{d}(s,t)$ $\quad$ \textcolor{gray}{// length of the shortest path that we have found so far ($\infty$ unless $s=t$)}\;
$u_s \gets s,\; u_t \gets t$ $\quad$ \textcolor{gray}{// vertices currently being explored in the two executions}\;
$i, j \gets \infty$\;
\While{$\hat d(s,u_s) + \hat d(u_t,t) < \mu$}{
    \textcolor{gray}{// relax one edge in forward search:}\;
    \While{$i\ge \outdeg(u_s)$}{
        \textbf{if} there is no forward-open vertex \textbf{then return} $\mu$\;
        $u_s \gets \arg\min \hat d(s,x)$ among all forward-open vertices $x$\;
        Mark $u_s$ as forward-closed\; 
        $i\gets 0$\;
        \textbf{if} $\hat d(s,u_s) + \hat d(u_t,t) \ge \mu$ \textbf{then return} $\mu$\;
    }
    $x\gets \outnbr(u_s, i)$, and mark $x$ as forward-open if it is not forward-closed\;
    $i\gets i+1$\;
        $\hat d(s,x) \gets \min\{\hat d(s,x), \; \hat d(s,u_s) + \ell(u_s, x)\}$ \;
    $\mu \gets \min\{ \mu, \; \hat d(s,x) + \hat d(x,t)\}$ \;
    
    \vspace{0.8em} \textcolor{gray}{// relax one edge in backward search: }\;
    \While{$j\ge \indeg(u_t)$}{
        \textbf{if} there is no backward-open vertex \textbf{then return} $\mu$\;
        $u_t \gets \arg\min \hat d(y,t)$ among all backward-open vertices $y$\;
        Mark $u_t$ as backward-closed \; 
        $j\gets 0$\;
        \textbf{if} $\hat d(s,u_s) + \hat d(u_t,t) \ge \mu$ \textbf{then return} $\mu$\;
    }
    $y\gets \innbr(u_t, j)$, and mark $y$ as backward-open if it is not backward-closed\;
    $j\gets j+1$\;
        $\hat d(y,t) \gets \min\{\hat d(y, t), \; \hat d(u_t, t) + \ell(y, u_t)\}$ \;
    $\mu \gets \min\{ \mu, \; \hat d(s,y) + \hat d(y,t)\}$ \;
}
\Return $\mu$
}
\end{algorithm}

\subsection{Fixed Minor Error From {\cite{sosa_bidirectional_HaeuplerHRTT25}}}
Note that the bidirectional Dijkstra algorithm presented in~\Cref{alg:bi-dj-new} differs slightly from the one in~\cite{sosa_bidirectional_HaeuplerHRTT25}. Specifically, the pseudocode in~\cite{sosa_bidirectional_HaeuplerHRTT25} updates $\mu$ upon relaxing an edge $(u,v)$ only when both endpoints $u$ and $v$ have been closed in the forward and backward searches, respectively (see Line 25 of Algorithm 2 in~\cite{sosa_bidirectional_HaeuplerHRTT25}). For example, when relaxing an edge $(u_s,x)$ in the forward search, the vertex $u_s$ is already forward-closed, and the algorithm in~\cite{sosa_bidirectional_HaeuplerHRTT25} updates $\mu$ according to~\Cref{eqn:update-mu} only if $x$ is also backward-closed. In contrast, our algorithm in~\Cref{alg:bi-dj-new} performs this update regardless of whether $x$ is backward-closed. 

As we show next, the extra condition from~\cite{sosa_bidirectional_HaeuplerHRTT25} should not be there as it can invalidate the claimed instance-optimality guarantee. To see this, consider the instance in~\Cref{fig:counterexample}, where edge labels denote weights. Here $u_2$ has $\Theta(n)$ outgoing edges and $v_2$ has $\Theta(n)$ incoming edges, while no path connects $u_2$ to $v_2$. The shortest $s$--$t$ path is $s\to u_1\to v_1\to t$, of length $3$. 

\begin{figure}[!h]
\centering
        \includegraphics[width=0.4\linewidth]{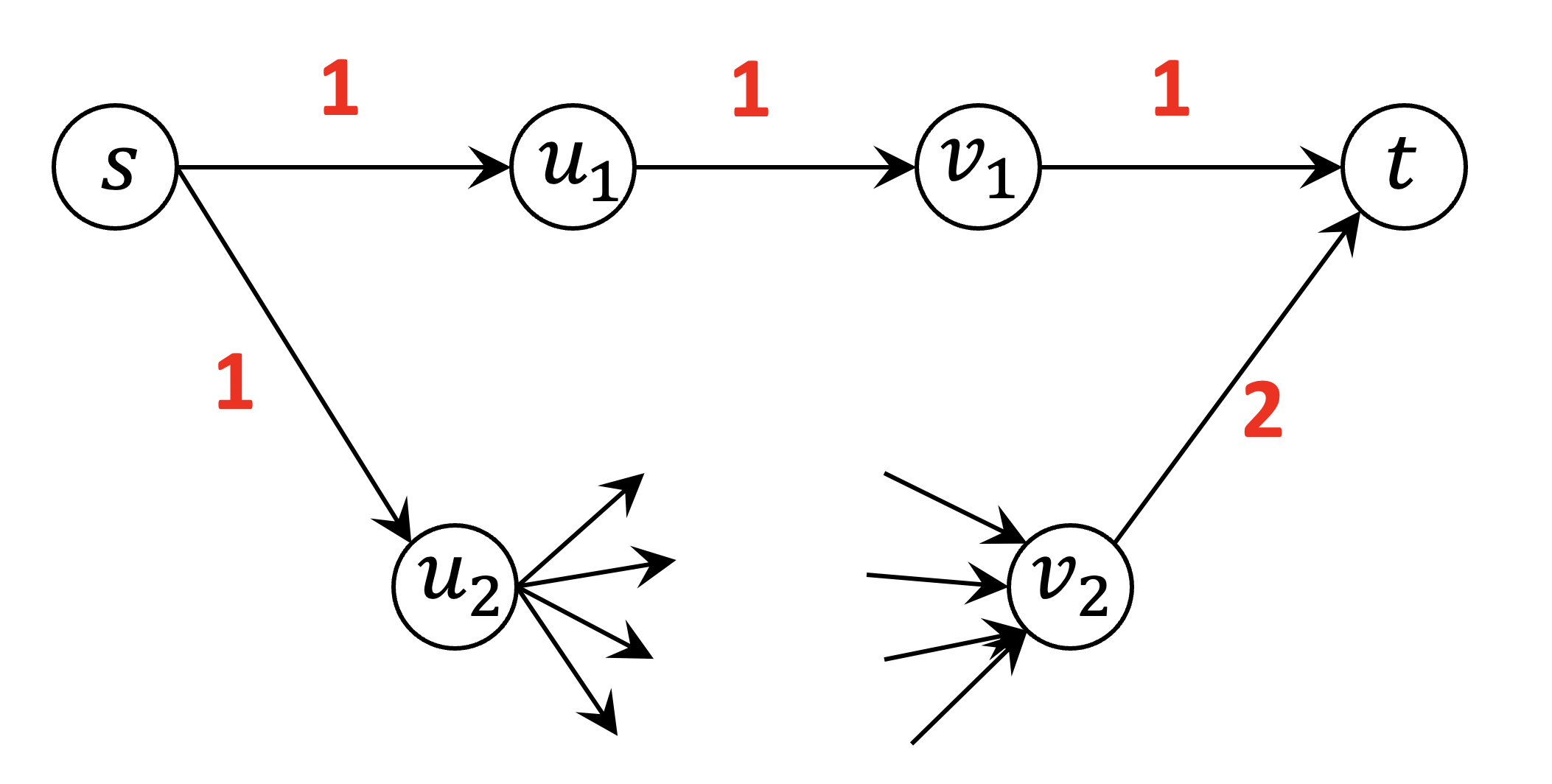}
        \vspace{-2mm}
        \caption{A counterexample arising from a minor error in~\cite{sosa_bidirectional_HaeuplerHRTT25}. The numbers on the edges denote their weights.}
        \label{fig:counterexample}
    \end{figure}

An instance-smart algorithm only needs to inspect the neighborhoods of $s$, $t$, and $v_1$ to certify that this path is optimal, and therefore terminates in $O(1)$ time. In contrast, under the update rule of~\cite{sosa_bidirectional_HaeuplerHRTT25}, after relaxing all outgoing edges of $s$, both $u_1$ and $u_2$ have tentative distance $1$, so either vertex may be selected as the next forward-closed vertex. If $u_2$ is selected, then when the edge $(u_1,v_1)$ is relaxed in the backward search, $\mu$ cannot be updated because $u_1$ has not yet been forward-closed. As a result, the algorithm must relax all $\Theta(n)$ outgoing edges of $u_2$ before $\mu$ can be updated and the stopping condition can be satisfied, resulting in $\Theta(n)$ time on this instance.

Technically speaking, the proof of the instance-optimality of bidirectional Dijkstra on weighted multigraphs (Theorem 5.1 in~\cite{sosa_bidirectional_HaeuplerHRTT25}) implicitly relies on the property that $d(s,u_1)+d(v_2,t)<d(s,t)$ holds for every edge $(u_1, v_1)$ relaxed by the forward search and every $(u_2, v_2)$ relaxed by backward search. However, this property does not hold under the update rule of~\cite{sosa_bidirectional_HaeuplerHRTT25}. We will discuss this with more details later in \Cref{lem:dist-st}. 

This minor error is resolved by removing the additional condition that requires $x$ to be backward-closed before $\mu$ is updated, as shown in~\Cref{alg:bi-dj-new}. Below we will first prove the correctness of our algorithm, and then show that it satisfies the above property. This fixes the minor error in the proof of~\cite{sosa_bidirectional_HaeuplerHRTT25} and establishes its claimed instance-optimality on weighted multigraphs.

\subsection{Correctness of Bidirectional Dijkstra}
The above-mentioned minor error introduces only unnecessary running time but does not affect the correctness of bidirectional Dijkstra. The correctness proof is essentially identical to that in~\cite{sosa_bidirectional_HaeuplerHRTT25}. For completeness, we briefly present it below.

Recall that $d(s,t)$ denotes the true shortest-path distance from $s$ to $t$. We assume without loss of generality that $0<d(s,t)<\infty$. Otherwise, either $s=t$ or no path exists from $s$ to $t$, in which case the correctness of \Cref{alg:bi-dj-new} is immediate.

We aim to show that \Cref{alg:bi-dj-new} returns $\mu=d(s,t)$. 
At the termination of \Cref{alg:bi-dj-new}, we have $\hat{d}(s, u_s)+\hat{d}(u_t,t)\ge \mu$. Since $u_s$ and $u_t$ are forward- and backward-closed, respectively, ~\Cref{lem:invariant_dijkstra}(a) gives
\begin{align}\label{eqn:usut}
\hat{d}(s, u_s)=d(s, u_s), \; \hat{d}(u_t, t)=d(u_t,t), \; \text{ and  }\; d(s,u_s)+d(u_t, t)\ge \mu. 
\end{align} 

We assume for contradiction that $d(s,t)<\mu$, and fix a shortest $s$--$t$ path $P$. 
Let $u$ be the last vertex on $P$ with $d(s,u)<d(s,u_s)$, and let $v$ be the successor of $u$ on $P$. By \Cref{lem:invariant_dijkstra}(b), $u$ is forward-closed. Moreover, $d(s,u_s)\le d(s,v)$, which together with~\Cref{eqn:usut} implies
\begin{align}\label{eqn:mu-upper}
\mu-d(u_t, t)\le d(s,u_s) \le d(s,v). 
\end{align}
Since $v$ lies on the shortest path $P$, $d(s,v)=d(s,t)-d(v,t)<\mu-d(v,t)$, and thus~\Cref{eqn:mu-upper} yields $d(v,t)<d(u_t, t)$. By~\Cref{lem:invariant_dijkstra}(b), $v$ is backward-closed.

Since $u$ and $v$ are forward-closed and backward-closed respectively before termination, the edge $(u,v)$ has been relaxed by both the forward search (when $u$ is closed) and by the backward search (when $v$ is closed), both before termination. At the later of these two events, both $\hat d(s,u)=d(s,u)$ and $\hat d(v,t)=d(v,t)$ hold, so the update rule for $\mu$ in~\Cref{alg:bi-dj-new} yields $\mu\le d(s,u)+\ell(u,v)+d(v,t)=d(s,t)<\mu$, a contradiction.

\subsection{Properties of Bidirectional Dijkstra}\label{subsec:property_bj}
We now summarize some known properties of bidirectional Dijkstra that will be used later. 

\paragraph{Visiting vertices and edges.}
We say that an edge $(u,v)$ is \emph{visited} by an algorithm if, during its execution, either $\outnbr$ or $\innbr$ is invoked at $u$ or $v$ and returns the other. We say that a vertex $v$ is \emph{visited} by an algorithm if, during its execution, either $\indeg(v)$ or $\outdeg(v)$ is invoked, or if at least one incident edge of $v$ is visited. Since each query causes $\Theta(1)$ vertices and edges to be visited, the total number of vertices and edges visited by any algorithm is asymptotically equal to the total number of query operations the algorithm performs.

\paragraph{Edge sets $\boldsymbol{E_s}$ and $\boldsymbol{E_t}$ of bidirectional Dijkstra.}
We now define $E_s$ and $E_t$ to be the sets of edges visited by the forward and backward searches, respectively, in the bidirectional Dijkstra's algorithm shown in~\Cref{alg:bi-dj-new}.
For undirected graphs, we define $E_s$ to be the set of ordered pairs of vertices $(u,v)$ such that the edge $(u,v)$ is visited by the forward search at $u$, and define $E_t$ symmetrically. We have

\begin{lemma}\label{lem:time_bj}
$|E_t|\leq |E_s|\leq |E_t|+1$. 
The number of queries performed by~\Cref{alg:bi-dj-new}, as well as the total number of visited vertices and edges, is $\Theta(|E_s|)=\Theta(|E_t|)$. 
\end{lemma}

\begin{proof}
\Cref{alg:bi-dj-new} alternates edge relaxations between the forward and backward searches, starting with the forward search. Therefore, $|E_s|-|E_t|$ is initially $0$ and alternates between $1$ (after a forward relaxation) and $0$ (after a backward relaxation) throughout the execution. 
Consequently, $|E_s|$ and $|E_t|$ differ by at most one. 

Consider the number of vertex and edge queries visited by~\Cref{alg:bi-dj-new}. A vertex will be considered visited the first time it is marked open, and (except for $s$ and $t$) a vertex is marked as open only after one of its incident edges has been relaxed. Hence, the number of visited vertices is asymptotically no larger than the number of relaxed edges. Therefore, the total number of visited vertices and edges is  $\Theta(|E_s|+|E_t|)=\Theta(|E_s|)=\Theta(|E_t|)$. Since by definition each visited vertex or edge incurs only $O(1)$ query operations, the total number of queries performed by \Cref{alg:bi-dj-new} is also $\Theta(|E_s|)=\Theta(|E_t|)$. 
\end{proof}

\paragraph{Vertex sets $\boldsymbol{V_s}$ and $\boldsymbol{V_t}$ of bidirectional Dijkstra.} Moreover, we define $V_s$ to be the set of forward-closed vertices from which at least one outgoing edge has been visited by the forward search, and $V_t$ to be the set of backward-closed vertices into which at least one incoming edge has been visited by the backward search. Note that the last vertex selected as $u_s$ or $u_t$, which triggers the termination criterion (i.e., Line 14 or Line 25 of~\Cref{alg:bi-dj-new}), is not included in $V_s$ or $V_t$, since the termination criterion is checked before any edge incident to the newly selected vertex is visited. 
By definition, for any edge $(u,v)\in E_s$, we have $u\in V_s$; analogously, for any edge $(u,v)\in E_t$, we have $v\in V_t$. 
\begin{lemma}
\label{lem:dist-st}
    For any $u\in V_s$ and $v\in V_t$, $d(s,u)+d(v,t)<d(s,t)$.
\end{lemma}

\begin{proof}
    Let $X$ denote the event that $\hat{d}(s,u_s)+\hat{d}(u_t,t)\ge \mu$ happens at the first time.
    Consider the moment that $X$ happens.
    Clearly, $X$ can only happen when we update $\mu$ (in Lines 18 or 29), or when we update $u_s$ or $u_t$ (in Lines 11 or 23). 
    By symmetry, we assume that it happens in either Line 18 or Line 11.
    
    In the latter case, the algorithm will terminate immediately before relaxing any other edges, so $u_s\notin V_s$.
    Let $u_s'$ denote the old $u_s$ right before $X$ happens in Line 11, then $d(s,u_s')+d(u_t,t)<\mu$.
    By \cref{lem:invariant_dijkstra}, $d(s,u)\le d(s,u_s')$ for any $u\in V_s$, and $d(v,t)\le d(u_t,t)$ for any $v\in V_t$.
    By the correctness of bidirectional Dijkstra, $\mu=d(s,t)$.
    Therefore $d(s,u)+d(v,t)<d(s,t)$ for any $u\in V_s$ and $v\in V_t$.

    It remains to prove that $X$ cannot happen in Line 18. Assume for contradiction that it happens in Line 18. When it happens, $\mu$ must be strictly decreased, so 
    $\mu=\hat{d}(s,x)+\hat{d}(x,t) \le \hat{d}(s,u_s)+\hat{d}(u_t,t)$. 
    We also have $\hat{d}(s,x)=\hat{d}(s,u_s)+\ell(u_s,x)$.
    This is because, 
    for any vertex $w\in V$, whenever we update $\hat{d}(s,w)$ or $\hat{d}(w,t)$, we will immediately update $\mu\gets\min\{\mu,\hat{d}(s,w)+\hat{d}(w,t)\}$.\footnote{The algorithm in \cite{sosa_bidirectional_HaeuplerHRTT25} does not satisfy this property due to the additional condition on updating $\mu$.}
    Therefore, we have
    \begin{align}\label{eqn:compare_mu}
    \mu=\hat{d}(s,u_s)+\ell(u_s,x)+\hat{d}(x,t) \le \hat{d}(s,u_s)+\hat{d}(u_t,t). 
    \end{align}
    Since $\ell(u_s,x)>0$,\footnote{We remark that this is the only place where this proof does not extend to graphs with zero-weight edges.} we also have $\hat{d}(x,t) \le \hat{d}(u_t,t)-\ell(u_s,x) < \hat{d}(u_t,t)$. 

    At this point, $u_t$ is backward-closed, so by~\Cref{lem:invariant_dijkstra}(b), $x$ was already backward-closed and all incoming edges backward-relaxed.
    In particular, $\hat{d}(u_s,t) \leq \ell(u_s,x)+\hat{d}(x,t)$.
    Again, since the algorithm updates $\mu$ whenever either $\hat{d}(s,w)$ or $\hat{d}(w,t)$ is updated for any vertex $w\in V$, it follows that before Line 18 we already had
    \(
        \mu
        \le \hat{d}(s,u_s)+\hat{d}(u_s,t)
        \le \hat{d}(s,u_s)+\ell(u_s,x)+\hat{d}(x,t)
        = \hat{d}(s,x)+\hat{d}(x,t).
    \)
    Thus Line 18 could not have strictly decreased $\mu$, a contradiction.
\end{proof}

%% file: technicaloverview.tex
\section{Setting Up the Analysis Framework}\label{sec:framework}

In this section, we set up the framework for establishing our upper and lower bounds on the instance-optimality ratio as shown in~\Cref{tab:results}. 
Recall from~\Cref{eqn:ratio} that the instance-optimality ratio of an algorithm $A$ is defined as 
\begin{align}\label{eqn:ratio_again}
\ratio_A = \sup_{I\in \cI}\{\cost_A(I)/\tau(I)\}, 
\end{align}
where $\cost_A(I)$ denotes the expected number of queries made by $A$ on instance $I$, and $\tau=\tau(I)$ is the \emph{instance complexity} of instance $I$, defined as the minimum value of $\cost_A(I)$ over all good algorithms $A$.

\subsection{Lower Bounds}\label{subsec:high-level_lower}
To establish an $\Omega(\ratio)$ lower bound on the instance-optimality ratio of bidirectional Dijkstra, we construct counterexamples exhibiting instances on which an instance-smart algorithm solves the shortest-path problem faster than bidirectional Dijkstra by at least a factor of $\ratio$.

To establish an $\Omega(\ratio)$ lower bound on the instance-optimality ratio of all algorithms, we show that for \emph{every} good algorithm $A$, there exists an instance on which an instance-smart algorithm performs significantly faster than $A$ by at least a factor of $\ratio$.

\subsection{Upper Bounds}
In~\Cref{lem:time_bj}, we established a tight bound on the expected number of queries made by bidirectional Dijkstra. In the following sections, we will establish lower bounds on the \emph{instance complexity}, thereby obtaining upper bounds on the instance-optimality ratio of bidirectional Dijkstra.

To construct the lower bound on the instance complexity, we for now fix a valid graph instance $G$ and an arbitrary good algorithm $A$. The algorithm $A$ may be randomized, with its randomness determined by a random seed $R$. But once the random seed $R$ is fixed, the execution of $A$ on $G$ is completely determined by the answers it receives from the graph oracle. Let $A_R$ denote the execution of $A$ with the random seed $R$. If there exist two graphs $G$ and $G’$ on which $A_R$ receives the same answers to all its queries, then $A_R$ produces exactly the same output on $G$ and $G'$.

The main idea of our proof is to construct another graph $G'$ whose structure differs from that of $G$ only on a small fraction of the graph (typically just a constant number of vertices and edges), which $A$ avoids. At the same time, $G'$ is carefully constructed so that the shortest-path distance from $s$ to $t$ differs between $G$ and $G'$.

Specifically, we use the expected number of vertices and edges visited by bidirectional Dijkstra as the basis of our analysis. Recall from~\cref{subsec:property_bj} that each query operation visits only $O(1)$ vertices and edges. Therefore, the number of vertices and edges visited by bidirectional Dijkstra is asymptotically equivalent to the number of queries it performs, namely $\Theta(|E_s|)=\Theta(|E_t|)$ as shown in~\Cref{lem:time_bj}. 

We assume for contradiction that $A$, as a good algorithm, visits significantly fewer vertices and edges than bidirectional Dijkstra. We then prove that, with constant probability, $A$ does not visit the region where $G'$ differs from $G$. Here, the probability is taken over the choice of the random seed used by $A$. 
Consequently, with constant probability, all query answers received by $A$ are identical on $G$ and $G'$, implying that $A$ will incorrectly produce the same shortest-path output on $G$ and $G'$ with constant probability. This contradicts the assumption on the goodness of $A$. 

To formalize the above idea, we define $p_{A,G}(u,v)$ as the probability that $A$ visits edge $(u,v)$ on $G$, and $p_{A,G}(v)$ as the probability that $A$ visits vertex $v$ on $G$. Moreover, for any $q\in [0,1]$, we say an edge $(u,v)$ is \emph{$q$-bad to $A$} (or simply \emph{$q$-bad} or \emph{bad} if the context is clear; similarly below) if $p_{A}(u,v)< q$; otherwise it is \emph{$q$-good}. Analogously, we say a vertex $v$ is \emph{$q$-bad} if $p_{A}(v)< q$; otherwise it is \emph{$q$-good}. For ease of presentation, we say that a vertex or an edge is \emph{$O(1)$-bad} if it is $q$-bad for some constant $q\in [0,1]$.

Note that in the order-dependent model, the probability is taken over the random seed used by $A$. In the order-oblivious model, the probability is additionally taken over the ordering of the vertices returned by the graph oracle in response to the neighbor queries $\innbr$ and $\outnbr$, where the ordering is chosen uniformly at random. We therefore have the following lemma. 
\begin{lemma}\label{lem:oblivious-edge-prob}\label{lem:oblivious-edge}
In the order-oblivious model, if an edge $(u,v)$ is $p$-good, then either every out-edge of $u$ is $q/4$-good or every in-edge of $v$ is $q/4$-good. 
\end{lemma}
\begin{proof}
    The edge $e=(u,v)$ is queried either as an out-edge of $u$ or as an in-edge of $v$. Suppose without loss of generality that $e$ is queried as an out-edge of $u$ with probability at least $q/2$.

    Let $\eq$ be the set of edges queried as out-edges of $u$, and fix another out-edge $e'\neq e$ of $u$.
    Fix the randomness of the algorithm and some order of the adjacency lists.
    If $e \in \eq$ but $e' \notin \eq$, then swapping $e$ and $e'$ in the out-going adjacency list of $u$ results in $e'$ being queried, since the algorithm behaves identically until this query.
    Since all orderings are equally likely, averaging over all randomness gives $\P{e' \in \eq} \geq \P{e \in \eq \wedge e' \notin \eq}$.
    Therefore,
    \begin{align*}
      \frac q2
      &\leq \P{e \in \eq}\\
      &= \P{e \in \eq \wedge e' \in \eq} + \P{e \in \eq \wedge e' \notin \eq}\\
      &\leq 2\P{e' \in \eq}.
    \end{align*}
    Since $e'$ was arbitrary, every out-edge of $u$ is queried with probability at least $q/4$.
\end{proof}

\begin{lemma}\label{lem:oblivious-edge-bad}
In the order-oblivious model, an edge $(u,v)$ is $4q$-bad if $u$ has at least one $q$-bad out-edge and $v$ has at least one $q$-bad in-edge.
\end{lemma}

\begin{proof}
\Cref{lem:oblivious-edge-bad} follows directly from~\Cref{lem:oblivious-edge}. Indeed, suppose for contradiction that edge $(u,v)$ is $4q$-good. Then, by~\Cref{lem:oblivious-edge}, either every out-edge of $u$ is $q$-good or every in-edge of $v$ is $q$-good, yielding a contradiction.
\end{proof}

We further define $Q_{A,G}$ as the expected number of vertices and edges in $G$ visited by $A$, that is, 
\begin{align}\label{eqn:querynumber}
Q_{A,G}=\sum_{(u,v)\in E}p_{A,G}(u,v)+\sum_{v\in V}p_{A,G}(v). 
\end{align}
We may omit $G$ or both $A$ and $G$ in the subscripts when the context is clear.

We construct $G'$ in one of the following ways, under the assumption that $A$ visits a significantly smaller total of vertices and edges than bidirectional Dijkstra when solving the shortest-path problem. Let $d_{G}(s,t)$ and $d_{G'}(s,t)$ denote the lengths of the shortest $s$-$t$ path in $G$ and $G'$, respectively. In each construction, we ensure that $d_{G’}(s,t)<d_G(s,t)$, and that $G'$ differs from $G$ only by (i) removing a constant number of $O(1)$-bad edges, (ii) adding a constant number of new edges, and thereby (iii) changing the degrees of a constant number of $O(1)$-bad vertices. By the union bound, the probability that $A$ visits at least one vertex or edge where $G'$ differs from $G$ is at most a constant. Consequently, $A$ receives exactly the same query answers on both $G$ and $G'$ with constant probability (e.g., $1/2$ as we use in our proof). This contradicts the assumption that $A$ is a good algorithm that outputs the correct answer on every instance with probability at least $9/10$. 

\paragraph{Construction (a).}\phantomsection \label{cons:a}
We construct $G’$ by finding an $O(1)$-length path $P$ from a vertex $u\in V_s$ to a vertex $v\in V_t$ in $G$ such that every edge on $P$ is $O(1)$-bad. Recall that \Cref{lem:dist-st} shows that $d(s,u)+d(v,t)<d(s,t)$. We then decrease the edge weights along $P$ to obtain $G’$ such that $d_{G’}(s,t)<d_G(s,t)$. In this construction, $G$ and $G’$ differ only on a few $O(1)$-bad edges on the path $P$, and the degrees of all vertices remain unchanged. 

\paragraph{Construction (b).} \phantomsection \label{cons:b}
There may be no suitable path $P$ in $G$, in which case we construct one by swapping edges. Specifically, we find an outgoing edge $(u_1,v_1)$ with $u_1\in V_s$ and an incoming edge $(u_2,v_2)$ with $v_2\in V_t$, such that both $(u_1,v_1)$ and $(u_2,v_2)$ are $O(1)$-bad edges. We then replace $(u_1,v_1)$ and $(u_2,v_2)$ with $(u_1,v_2)$ and $(u_2,v_1)$ while preserving the degrees of all vertices. This is possible provided that neither $(u_1,v_2)$ nor $(u_2,v_1)$ already exists in $G$. Note that $(u_1,v_2)$ certainly does not exist; otherwise, we could simply apply Construction~(\hyperref[cons:a]{a}). Such an edge swap creates a one-hop path from $u_1\in V_s$ to $v_2\in V_t$, while modifying only two bad edges, adding two new edges, and preserving the degrees of all vertices. Moreover, on directed graph, it does not change the shortest-path distance from $s$ to any vertex in $V_s$, or from any vertex in $V_t$ to $t$. We then apply Construction~(\hyperref[cons:a]{a}) to obtain $G'$ such that $d_{G'}(s,t)<d_G(s,t)$.
    
\paragraph{Construction (c).}\phantomsection \label{cons:c}
We may also modify bad vertices. Specifically, we select a bad vertex $w$ and find an outgoing edge $(u_1,v_1)$ with $u_1\in V_s$ and an incoming edge $(u_2,v_2)$ with $v_2\in V_t$, such that both $(u_1,v_1)$ and $(u_2,v_2)$ are $O(1)$-bad edges. 
We remove the edges $(u_1,v_1)$ and $(u_2,v_2)$, and add the edges $(u_1,w)$ and $(w,v_2)$. We further add edges incident to $v_1$ and $u_2$ from other bad vertices so as to preserve the degrees of all vertices except these bad vertices. 
We then apply Construction~(\hyperref[cons:a]{a}) to construct $G’$ such that $d_{G’}(s,t)<d_G(s,t)$. 

%% file: undirected_oblivious.tex
\section{Order-Oblivious Instance-Optimality for Undirected Graphs}\label{sec:order-oblivious-undirected}
In this section, we show that bidirectional Dijkstra is instance-optimal for undirected graphs in the order-oblivious model, thus positively answering the open question left by~\cite{sosa_bidirectional_HaeuplerHRTT25} regarding the instance-optimality of bidirectional Dijkstra on simple weighted graphs. 

\begin{theorem}
\label{thm:undirected-oblivious-undirected}
In the order-oblivious model, \Cref{alg:bi-dj-new} is instance-optimal on undirected graphs. 
\end{theorem}

\begin{proof}
    We adopt the Construction~(\hyperref[cons:a]{a}) and~(\hyperref[cons:b]{b}) from~\Cref{sec:framework} to construct $G'$. 

    Recall from~\Cref{lem:time_bj} that the total number of vertices and edges visited by \Cref{alg:bi-dj-new} is $\Theta(|E_s|)=\Theta(|E_t|)$, where $E_s$ and $|E_t|$ denote the set of edges visited by the forward and backward searches in~\Cref{alg:bi-dj-new}, respectively. 
    Assume for contradiction that $Q_A<|E_t|/20$, where we recall from~\Cref{eqn:querynumber} that $Q_{A}$ denotes the expected number of vertices and edges visited by $A$. 
    Since a $q$-bad edge $(u,v)$ satisfies $p_{A,G}(u,v)<q$, the assumption $Q_A<|E_t|/20$ implies that there exist at least one $1/20$-bad edge $(u_1,v_1)\in E_s$ and at least one $1/20$-bad edge $(u_2,v_2)\in E_t$.
    
    If $u_1,v_1,u_2,v_2$ are not all distinct, then since this is for undirected graphs, $u_1$ and $v_2$ are already connected by at most two $1/20$-bad edges, completing Construction~(\hyperref[cons:a]{a}). 
    
    Otherwise, $u_1,v_1,u_2,v_2$ are all distinct. Note we are considering the order-oblivious model, in which the orders of the vertices and the incident edges are chosen uniformly at random. By~\Cref{lem:oblivious-edge-bad}, if there exists an edge between a vertex in $\{u_1,v_1\}$ and a vertex in $\{u_2,v_2\}$, then this edge must be $1/5$-bad since each of its endpoints has an incident $1/20$-bad edge.
    Hence, $u_1$ and $v_2$ are connected by one $1/5$-bad edge and at most two $1/20$-bad edges.
    Otherwise, there is no edge between ${u_1,v_1}$ and ${u_2,v_2}$. We then replace $\{(u_1,v_1),(u_2,v_2)\}$ by $\{(u_1,v_2),(u_2,v_1)\}$, completing Construction~(\hyperref[cons:b]{b}).
   \end{proof}

We remark that \Cref{thm:undirected-oblivious-undirected} shows the only setting in which bidirectional Dijkstra is instance-optimal over all valid instances. In all other settings, it is off by a super-polylogarithmic factor, as we will see in the following sections.

\section{Order-Dependent Instance-Non-Optimality for Directed Graphs}\label{sec:order-dependent-directed}

We now begin presenting our lower bounds on the instance-optimality ratio. For ease of exposition, we present our results in technical order. 
In this section, we focus on directed graphs in the order-dependent model. The lower bounds for directed graphs in the order-oblivious model and for undirected graphs in the order-dependent model will be given later in \Cref{sec:order-oblivious-directed} and~\Cref{sec:order-dependent-undirected}, respectively.  

Recall that in the order-dependent model, the orders of the vertices and their incident edges are part of the input and may therefore be exploited by an instance-smart algorithm. However, an algorithm that aims to be instance-optimal (e.g., bidirectional Dijkstra) must be the best under every possible ordering. We therefore regard the ordering as being chosen adversarially, and compare its query complexity against that of instance-smart algorithms tuned for this ordering.

We prove the following lower bound, when parameterizing in the number of vertices $n$ and the number of edges $m$.

\begin{theorem}
\label{thm:directed-fixed-lower-m}
    In the order-dependent model, the instance-optimality ratio of every algorithm is $\Omega(\lceil m/n \rceil)$ on directed graphs.
\end{theorem}

This result will essentially follow from the proof of the following lower bound, parameterized in the number of vertices $n$ and the maximum in- and out-degree $\Delta$.

\begin{theorem}
\label{thm:directed-fixed-lower}
    In the order-dependent model, the instance-optimality ratio of every algorithm is $\Omega(\lceil\Delta^2/n\rceil)$ on directed graphs.
\end{theorem}

In~\Cref{sec:more-positive}, we will show that the upper bound on the instance-optimality ratio of bidirectional Dijkstra matches the above lower bounds for every algorithm. 
Consequently, bidirectional Dijkstra is the best possible algorithm in terms of instance-optimality ratio on directed graphs in the order-dependent model.

\begin{proof}[Proof of~\Cref{thm:directed-fixed-lower}]
Recall from~\Cref{subsec:high-level_lower} that we establish lower bounds on the instance-optimality ratio of all algorithms by showing that for every good algorithm $A$, there exists an instance on which an instance-smart algorithm performs significantly faster than $A$.

    For any given $\Delta \in [0, n]$, consider the following graph instance which has $n$ vertices and maximum degree $\Theta(\Delta)$. 
    
    The graph consists of several components: $V=\{s,t\}\cup V_1\cup V_2\cup V_3\cup V_4 \cup V_5$, where $|V_i|=\Theta(\Delta)$ for $i\in[1,4]$,  and $|V_5|=n-\Theta(\Delta)$. The target vertex $t$ together with $V_5$ forms an outgoing path starting from $t$.
    Each vertex except $t$ and those in $V_5$ has an outgoing edge to every vertex in $V_1$, and each vertex except $s$ and those in $V_5$ has an incoming edge from every vertex in $V_2$. In addition, the source vertex $s$ has an outgoing edge to every vertex in $V_3$, each of weight $1$. The target vertex $t$ has an incoming edge from every vertex in $V_4$, each of weight $1$. Finally, there is a perfect matching from $V_3$ to $V_4$. A sketch for the above graph construction is given in~\cref{fig:directed-fixed}.
    
    \begin{figure}[!h]
        \centering
        \includegraphics[width=0.5\linewidth]{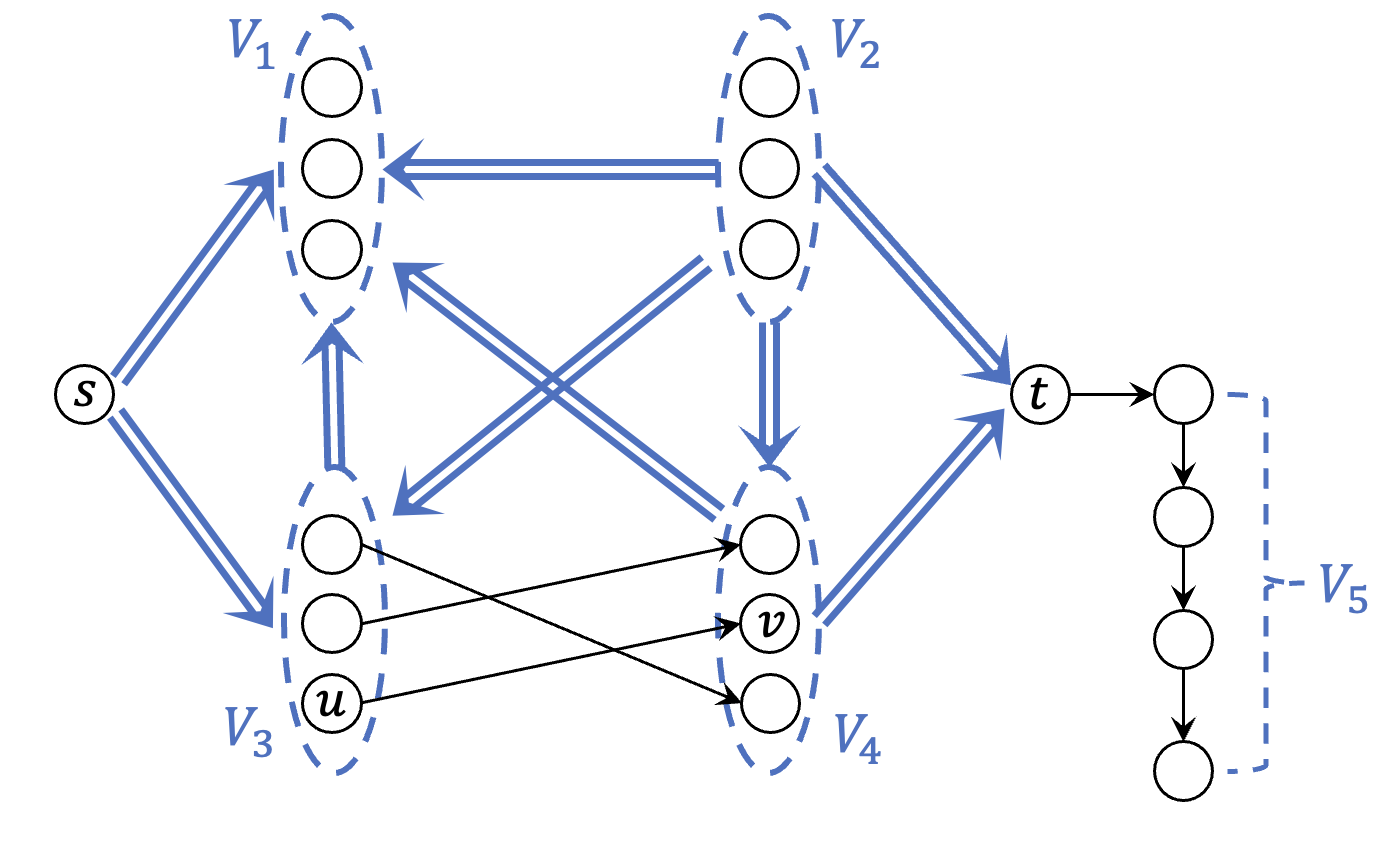}
        \vspace{-2mm}
        \caption{Lower-bound construction for directed graphs in the order-dependent model. A blue block arrow from a set $X$ to a set $Y$ indicates that every vertex in $X$ has a directed edge to every vertex in $Y$.}
        \label{fig:directed-fixed}
    \end{figure}

An instance-smart algorithm first visits all incident edges of $s$ and $t$. From this, it learns the in- and out-degrees of the vertices in $V_1,V_2, V_3,V_4$, as well as the existence of a path from $t$ to $V_5$. By visiting all incident edges of the vertices in $V_5$, the algorithm verifies that $V_5$ is disconnected from the remaining components of the graph.

From the out-degrees of the vertices in $V_1$ and the in-degrees of the vertices in $V_2$, the algorithm concludes that there are no outgoing edges from $V_1$ and that $t\notin V_1$. Similarly, there are no incoming edges to $V_2$ and $s\notin V_2$. Consequently, if there exists an $s$–$t$ path, then no vertex in $V_1\cup V_2$ can lie on any shortest $s$–$t$ path.

Moreover, by comparing the in-degree of each vertex in $V_1$ with the total number of vertices $n$, together with the fact that neither $t$ nor any vertex in $V_5$ has an edge to $V_1$, the algorithm concludes that every vertex other than $t$ and those in $V_5$ has an outgoing edge to every vertex in $V_1$. Analogously, every vertex other than $s$ and those in $V_5$ has an incoming edge from every vertex in $V_2$.
Combining this information with the out-degrees of the vertices in $V_3$, the algorithm concludes that, for each vertex $u\in V_3$, at most one outgoing edge of $u$ needs to be visited in order to solve the shortest-path problem. Let $(u,v)$ denote this edge, and let $i$ denote the position of $v$ in the outgoing adjacency list of $u$, that is, $\outnbr(u,i)=v$. The instance-smart algorithm tuned for this instance first queries $\outnbr(u,i)$. Since it has already visited all vertices in $V_1$, it can verify whether the returned vertex belongs to $V_1$. Once it confirms that $v\notin V_1$, it knows that this is the unique outgoing edge of $u$ outside $V_1$, and therefore no further out-neighbor queries from $u$ are necessary.

Consequently, the instance-smart algorithm visits only one outgoing edge of each vertex in $V_3$ and one incoming edge of each vertex in $V_4$, after which it can compute the shortest $s$–$t$ path. The total number of vertices and edges visited is therefore $O(n)$.

In contrast, for an algorithm that aims to be instance-optimal, as explained above, we should regard the order of the edges in the outgoing adjacency lists of the vertices in $V_3$ as being chosen adversarially. 
Consequently, for each vertex $u\in V_3$, the algorithm makes $\Omega(\dout(u))=\Omega(\Delta)$ queries in expectation before encountering the edge $(u,v)$, where the expectation is taken over the internal randomness of the algorithm.
A symmetric argument applies to the incoming adjacency lists of the vertices in $V_4$. Since every outgoing edge from $s$ to a vertex in $V_3$ and every incoming edge from a vertex in $V_4$ to $t$ has identical weight $1$, the algorithm must query all edges from $V_3$ to $V_4$ to ensure that none of them yields a shorter path.
This results in a total of $\Omega(\Delta^2)$ queries.
\end{proof}

A small modification of this proof gives us~\Cref{thm:directed-fixed-lower-m}.

\begin{proof}[Proof of~\Cref{thm:directed-fixed-lower-m}]
The lower-bound construction for~\Cref{thm:directed-fixed-lower} yields a graph with $\Theta(\Delta^2+n)$ edges. We now have a given $m\in[n,n^2]$. So we set $|V_i|=\Theta(\sqrt{m})$ for each $i\in[1,4]$, so that the resulting graph has $\Theta(m)$ edges. The proof then follows directly from the above construction.
\end{proof}

\section{Order-Oblivious Instance-Non-Optimality for Directed Graphs}\label{sec:order-oblivious-directed}

We move on to the setting of directed graphs in the order-oblivious model. We start with our lower bound for every algorithm.
When parameterizing in $n$ and $m$, we again show a lower bound of $\Omega(\lceil m/n \rceil)$ for every algorithm, but only when the average degree is $O(\sqrt n)$. We will later show an $\Omega(\lceil m/n \rceil)$ lower bound without the assumption on the average degree, when restricting attention to bidirectional Dijkstra.

\begin{theorem}
\label{thm:directed-oblivious-m-any}
    In the order-oblivious model, the instance-optimality ratio of every algorithm is $\Omega(\lceil m/n \rceil)$ on directed graphs with $m=O(n\sqrt n)$.
\end{theorem}

\begin{proof}
    The case $m<n$ is trivial, so we assume that $m\ge n$.
    Since $m=O(n\sqrt n)$, we have $m/n=O(\sqrt n)$.

    We construct the graph $G=(V,E)$ shown in~\cref{fig:directed-oblivious}.
    The graph has several components: $V = \{s,t, p', q'\} \cup V_1 \cup V_2 \cup V_3 \cup V_4 \cup V_5 \cup V_6 \cup V_7$.
    Let $|V_1|=|V_2|=|V_3|=|V_4|=|V_5|=|V_6|=\Theta(m/n)$ and $|V_7|=n-\Theta(m/n)$, (not pictured) so that $|V|=n$.
    Pick vertices $p \in V_1$, $x \in V_2$, $y \in V_3$, $q \in V_4$, $v^* \in V_5$, and $u^* \in V_6$.
    The vertices in $V_5\cup\{s\}$ form a path ending at $s$, and the vertices in $\{t\}\cup V_6$ form a path starting at $t$.
    Moreover, $s$ has an outgoing edge to every vertex in $V_2$, and every vertex in $V_3$ has an outgoing edge to $t$.
    Every vertex in $V\setminus(\{s,p',q',t\}\cup V_1 \cup V_6)$ has an edge to every vertex in $V_1$, apart from the edge $(v^*,p)$ which is left absent.
    Similarly, every vertex in $V_4$ has an edge to every vertex in $V\setminus(\{s,p',q',t\}\cup V_4 \cup V_5)$, apart from the edge $(q, u^*)$ which is left absent.
    Finally, edges $(v^*, p')$ and $(q', u^*)$ are present.
    The number of edges is $\Theta(m)$.
    To construct the modified graph $G'$ from $G$, remove (red dashed) edges $\{(v^*,p'), (x,p), (q,y), (q',u^*)\}$ and add (green) edges $\{(v^*,p), (x,y), (q,u^*)\}$.
    
    \begin{figure}[!h]
        \centering
        \includegraphics[width=0.6\linewidth]{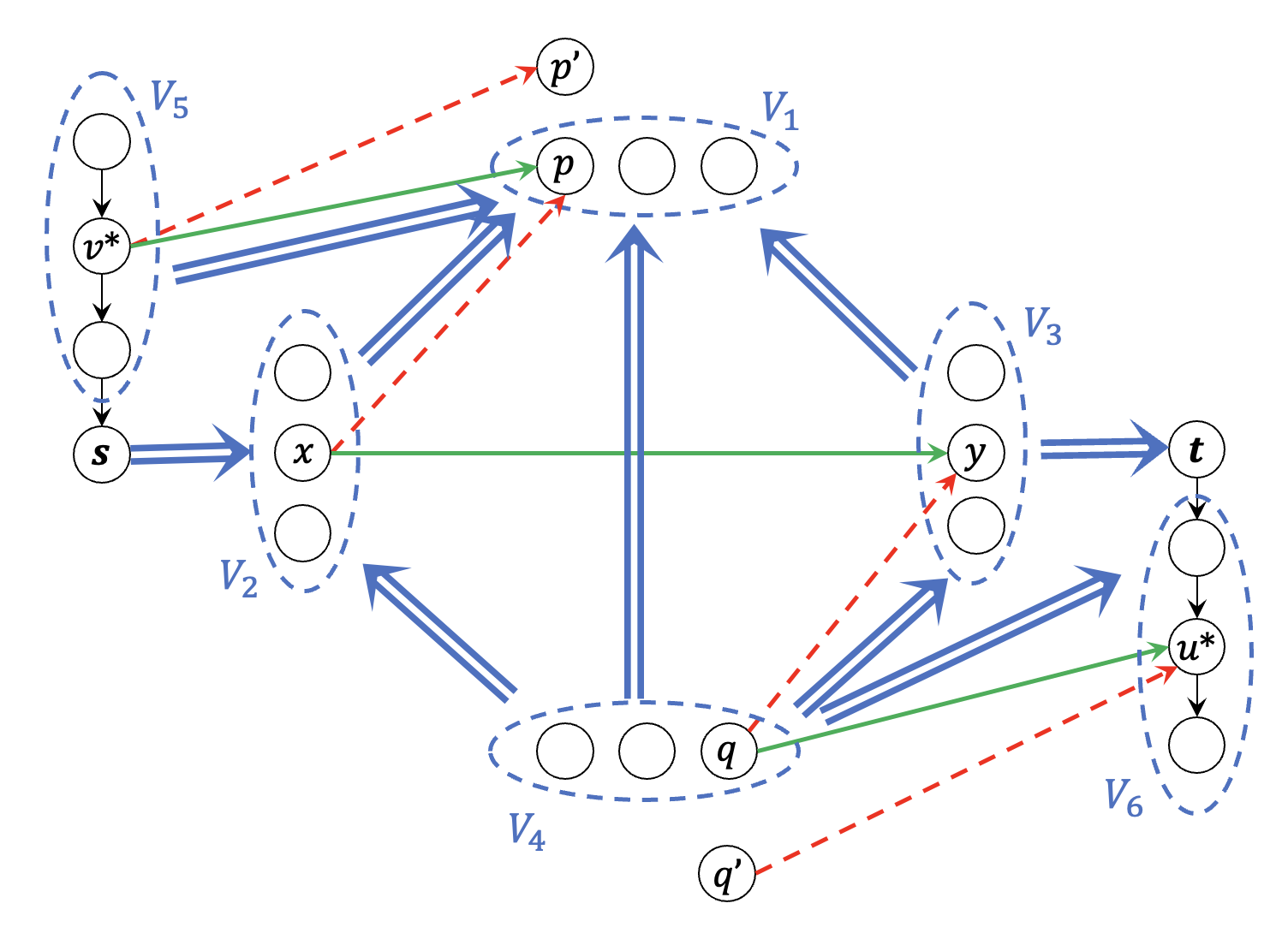}
        \vspace{-2mm}
        \caption{Lower-bound construction for directed graphs in the order-oblivious model. A blue block arrow from a set $X$ to a set $Y$ indicates that every vertex in $X$ has a directed edge to every vertex in $Y$, except that the edges $(v^*,p)$ and $(q,u^*)$ are absent from $G$. Red dashed edges are present in $G$ but absent from $G'$, while green edges are present in $G'$ but absent from $G$.} 
        \label{fig:directed-oblivious}
    \end{figure}

    An instance-smart algorithm can verify that $t$ is not reachable from $s$ in $G$ as follows.
    Identify $V_1$ as the set of out-neighbors of an arbitrary out-neighbor of $s$.
    Observe that every vertex in $V_1$ has in-degree $n-4-|V_1|-|V_6|$ with no in-edges from $\{s,t,p',q'\}\cup V_1 \cup V_6$, other than one vertex, $p$, which has in-degree $n-5-|V_1|-|V_6|$ due to the additional missing in-edge from $v^*$. 
    Walk to $v^*$ from $s$ by following the unique in-edge for a number of steps known to the instance-smart algorithm.
    Query every out-neighbor of $v^*$ to find, that it has no out-neighbor in $V_1$.
    Furthermore, one of these out-neighbors will be $p'$, identifiable as the one outside $V_1 \cup V_5 \cup \{s\}$.
    Similarly find $q'$.
    Observe that no out-edge from $\{s,p',q',t\}\cup V_1 \cup V_6$ goes to $V_1$, and recall that we have observed that $(v^*,p)$ is absent.
    Combining this with the degrees of the vertices in $V_1$ certifies, that every vertex outside $\{s,p',q',t, v^*\}\cup V_1 \cup V_6$ has $|V_1|$ edges to $V_1$.
    Finally, walk to every out-neighbor of $s$, i.e.\ every vertex in $V_2$, and observe that all out-degrees are $|V_1|$.
    This certifies that all out-edges of $V_2$ go to $V_1$, which we know has no out-edges by querying out-degrees.
    This certifies that no edge going out of $\{s\} \cup V_2 \cup V_1$ exists, and thus, that no path from $s$ to $t$ exists.
    This can all be done with instance-complexity $O(m/n)$. 

   Note that it is crucial that the instance-smart algorithm ``knows'' exactly how many steps to walk from $s$ along the path before reaching $v^*$ (or, symmetrically, from $t$ before reaching $u^*$). This number of steps acts as a kind of “lucky number” tailored to this particular instance. Regardless of whether the input graph is $G$ or not, the instance-smart algorithm can simply walk this many steps along the path and then check whether it encounters the edge $(v^*,p')$ (or $(q', u^*)$). On the instance $G$, this allows the instance-smart algorithm to run in only $O(m/n)$ queries.

    In contrast, an algorithm that aims to be instance-optimal on all instances cannot rely on such instance-specific ``lucky numbers.'' Knowing a number that is useful only for a single instance does not help improve the instance-optimality ratio, since the algorithm must, on every instance, achieve a running time no more than a factor of the instance-optimality ratio larger than that of every instance-smart algorithm tuned to the corresponding instance.
    As a result, such an algorithm needs to consider the choices of vertices $p,x,y,q,v^*,u^*$ from their respective sets as adversarial. This is essentially equivalent to assuming that $p,x,y,q,v^*,u^*$ are chosen uniformly at random from their respective sets.
    In the following, we show that any such algorithm requires $\Omega((m/n)^2)$ queries on this construction of $G$.

    Consider the execution of the algorithm on $G$, where the vertices $p,x,y,q,v^*,u^*$ are chosen uniformly at random from their respective sets and the vertex and edge orders are chosen uniformly at random.
    Call a query \emph{revealing} if it visits one of the four dashed red edges $\{(v^*,p'), (x,p), (q,y), (q',u^*)\}$, or finds one of the special vertices $p'$ and $q'$.
    If no revealing query is made, then the algorithm receives exactly the same answers on $G$ and on $G'$.
    
    We now bound the probability that a revealing query happens before time $c(m/n)^2$ for some small positive constant $c$.
    Condition on everything the algorithm has seen so far, assuming no revealing query has happened.
    Then $(x,p)$ is still hidden among $\Theta((m/n)^2)$ possible edges from $V_2$ to $V_1$, except for the positions already ruled out by previous queries.
    As long as fewer than $c(m/n)^2$ queries have been made, the next query therefore hits $(x,p)$ with probability $O((n/m)^2)$.
    The same argument applies to $(q,y)$.
    
    For the edge $(v^*,p')$, conditioned on not having found $p'$, the algorithm has to find both the special vertex $v^*\in V_5$ and the position of $p'$ in its outgoing adjacency list.
    Thus $(v^*,p')$ is again hidden among $\Theta((m/n)^2)$ possible remaining choices.
    The edge $(q',u^*)$ is symmetric.
    
    It remains to account for finding $p'$ or $q'$.
    Finding them by direct vertex queries takes $\Omega(n)$ queries in expectation.
    Since $m=O(n\sqrt n)$, we have $(m/n)^2=O(n)$, so this is also $\Omega((m/n)^2)$.
    Finding them through an incident edge is already covered by the same argument as above, other than as a neighbor of $V_1$ or $V_2$, but every vertex in $V_1$ and $V_2$ has relevant degree $\Omega(n)$, so this is no better.
    
    Therefore, at any time before $c(m/n)^2$ queries, conditioned on no earlier revealing query, the next query is revealing with probability $O((n/m)^2)$.
    Summing over the first $c(m/n)^2$ queries, the probability that the algorithm makes any revealing query is $O(c)$.
    Choosing $c>0$ small enough, this probability is less than $1/10$.
    
    With probability at least $9/10$, the algorithm therefore sees the same query answers on $G$ and on $G'$, and makes the same output.
    However, $G$ has no path from $s$ to $t$, while $G'$ does.
    This contradicts that the algorithm is good.
    So every good algorithm must make $\Omega((m/n)^2)$ queries in expectation on some instance in this family.
\end{proof}

Above, in~\Cref{thm:directed-oblivious-m-any}, we establish an $\Omega(\lceil m/n \rceil)$ lower bound on the instance-optimality ratio of every algorithm, under the assumption that $m/n=O(\sqrt{n})$. 
Now, restricting our attention to bidirectional Dijkstra, we show stronger lower bounds by removing the assumption. 

First, when parameterizing in $n$ and $m$, we get the following bound on bidirectional Dijkstra.

\begin{theorem}
\label{thm:directed-oblivious-m}
    In the order-oblivious model, the instance-optimality ratio of \Cref{alg:bi-dj-new} is $\Omega(\lceil m/n \rceil)$ on directed graphs.
\end{theorem}

As in the previous section, this result will follow from the proof of the following bound, which is parameterized in $n$ and the maximum in- and out-degree $\Delta$.

\begin{theorem}
\label{thm:directed-oblivious}
    In the order-oblivious model, the instance-optimality ratio of \Cref{alg:bi-dj-new} is $\Omega(\lceil\Delta^2/n\rceil)$ on directed graphs.
\end{theorem}

\begin{proof}
    Consider the graph in \cref{fig:directed-fixed} again. Now we replace the matching from $V_3$ to $V_4$ with only one edge $(u,v)$ for some $u\in V_3$ and $v\in V_4$. 
      
    The instance-smart algorithm can first query all in-edges and out-edges of vertices in $\{s,t\}\cup V_5$ and the in-degrees and out-degrees of all vertices. Then it can find $u$ and $v$ by their degrees and it also knows that it suffices to find the edge $(u,v)$. So it only needs to query all out-edges of $u$. The total time complexity is $O(n)$.

    On the other hand, bidirectional Dijkstra will query all out-edges of vertices in $V_3$ and all in-edges of vertices in $V_4$, so it spends $\Omega(\Delta^2)$ time.
\end{proof}

\Cref{thm:directed-oblivious-m} follows from the proof of~\Cref{thm:directed-oblivious}, by setting $|V_i|=\Theta(\sqrt m)$ for $i\in[1,4]$. 

Let us again mention that, 
in~\Cref{sec:more-positive}, we will give matching upper bounds on the instance-optimality ratio of bidirectional Dijkstra in both parameterizations, implying that the above lower bounds on the instance-optimality ratio of bidirectional Dijkstra are already tight.

\section{Order-Dependent Instance-Non-Optimality for Undirected Graphs}\label{sec:order-dependent-undirected}

The final setting is undirected graphs in the order-dependent model, where we establish lower bounds on the instance-optimality ratio for all algorithms. In particular, when parameterizing by $n$ and $m$, we obtain the following result.

\begin{theorem}
\label{thm:undirected-fixed-lower-m}
    In the order-dependent model, the instance-optimality ratio of every algorithm is $\Omega(\lceil m/n \rceil)$ on undirected graphs when $m<n^{4/3}$. \end{theorem}

This result will essentially follow from the proof of the following bound, which is parameterized in $n$ and the maximum in- and out-degree $\Delta$.

\begin{theorem}\label{thm:undirected-fixed-lower}
    In the order-dependent model, the instance-optimality ratio of every algorithm is $\Omega(\lceil\Delta^{2/3}/n^{1/3} \rceil)$ on undirected graphs.
\end{theorem}

In~\Cref{sec:more-positive}, we will show matching upper bounds for bidirectional Dijkstra, in both parameterizations.
Consequently, bidirectional Dijkstra is the best possible algorithm in terms of instance-optimality ratio on undirected graphs in the order-dependent model.

\begin{proof}[Proof of \Cref{thm:undirected-fixed-lower}]
    For any $\Delta\in[\sqrt{n},n]$, we construct a graph with $n$ vertices and maximum degree $\Theta(\Delta)$ as follows. 

    Analogous to the directed case, the graph consists of several components: $V=\{s,t\}\cup V_1\cup V_2\cup V_3\cup V_4$, where $|V_1|=\Theta(n^{1/3}\Delta^{1/3}),|V_2|=\Theta(n^{2/3}/\Delta^{1/3}),|V_3|=\Theta(\Delta)$, and $|V_4|=n-2-|V_1\cup V_2 \cup V_3|$ so that the number of vertices in the graph is $n$. The edges are defined as follows. The target vertex $t$ together with $V_4$ forms a path starting from $t$. The source vertex $s$ (respectively, the target vertex $t$) is connected to every vertex in $V_1$ (respectively, $V_2$) by an edge of weight $1$. Every vertex in $V_2$ is connected to every vertex in both $V_1$ and $V_3$. 
    In addition, every pair of vertices within $V_1$ and within $V_2$ is connected by an edge. 
    As a result, the degree of each vertex in $V_2$ is $n-2-|V_4|$, since it is adjacent to every vertex in the graph except itself and the vertices in $s\cup V_4$. 
    A sketch of the construction is shown in~\Cref{fig:undirected-fixed}.
    
     \begin{figure}[!h]
        \centering
        \includegraphics[width=0.6\linewidth]{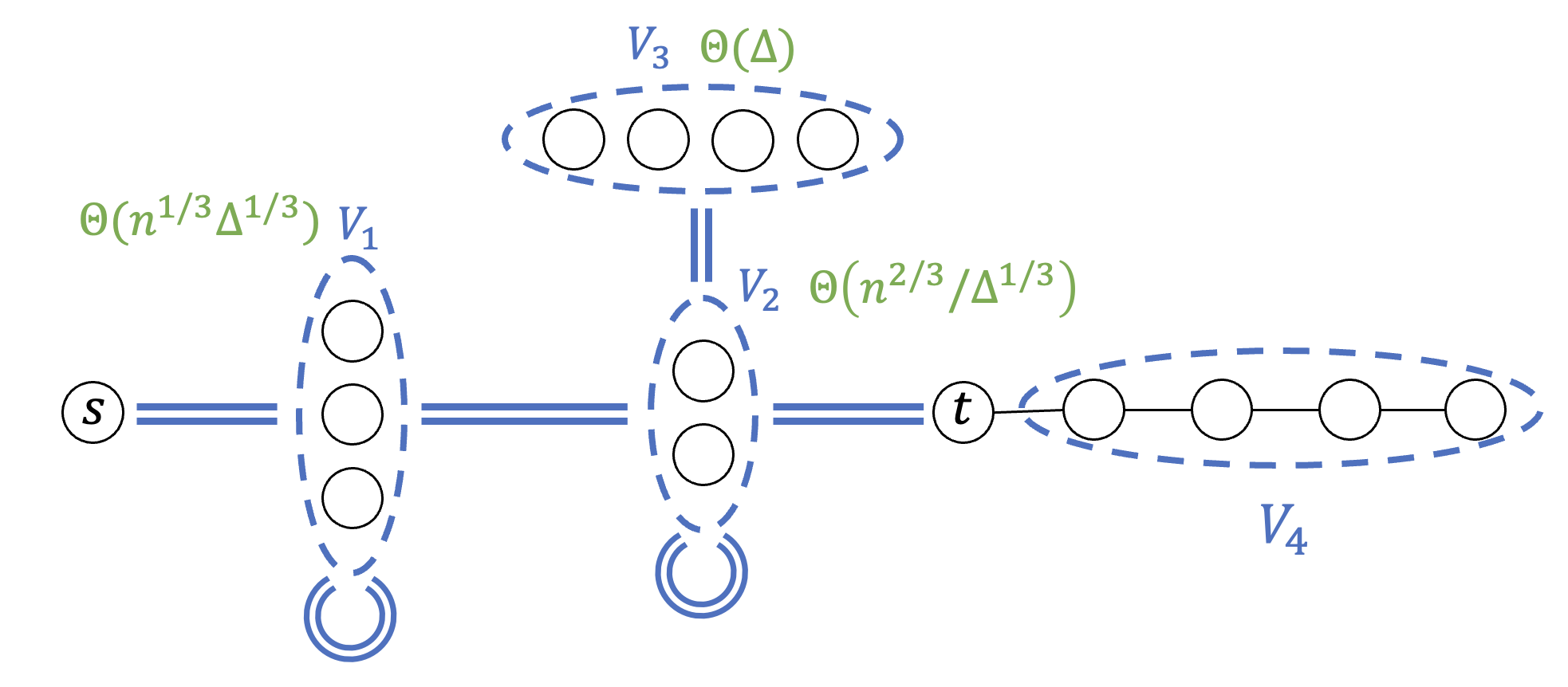}
        \vspace{-1mm}
        \caption{Lower-bound construction for undirected graphs in the order-dependent model. A blue block line between two vertex sets indicates that every vertex in one set is connected by an edge to every vertex in the other set. A circular blue block line around a vertex set indicates that every pair of vertices within the set is connected by an edge. 
        The green numbers indicate the sizes of the corresponding vertex sets.}
        \label{fig:undirected-fixed}
    \end{figure}

An instance-smart algorithm first checks all edges incident to $s$ and $t$, thereby identifying the sets $V_1$, $V_2$, and $V_4$. It then verifies all edges incident to the vertices in $V_4$, confirming that they are disconnected from the remaining vertices. It also confirms that $s$ is not adjacent to any vertex in $V_2$, and $t$ is not adjacent to any vertex in $V_1$.

Next, the algorithm queries the degree of each vertex in $V_2$. Comparing these degrees with the total number of vertices $n$, together with the fact that neither $s$ nor any vertex in $V_4$ is adjacent to $V_2$, it concludes that every vertex in $V_2$ is adjacent to every vertex in the graph except itself and the vertices in ${s}\cup V_4$. 

While examining the edges incident to $t$, the algorithm also observes that all of them, except the edge to $V_4$, have the same weight. Since every neighbor of $t$ belongs to either $V_2$ or $V_4$, and no path from $t$ to $s$ can pass through $V_4$, it concludes that the edge weights within $V_2$ are irrelevant to the $s$–$t$ distance.

The algorithm then queries the degree of each vertex in $V_3$. Comparing these degrees with $|V_2|$, and using the fact that every vertex in $V_2$ is adjacent to every vertex in the graph except $t$ and the vertices in $V_4$, it concludes that every vertex in $V_3$ is adjacent only to vertices in $V_2$. Hence, the edge weights of the edges between $V_2$ and $V_3$ are likewise irrelevant to the $s$–$t$ distance.

Next, the algorithm queries the degree of each vertex in $V_1$. Comparing these degrees with the total number of vertices $n$, together with the fact that no vertex in ${t}\cup V_3\cup V_4$ is adjacent to $V_1$, it concludes that every vertex in $V_1$ is adjacent to every vertex in ${s}\cup V_2$, as well as to every other vertex in $V_1$.

Finally, from the edges incident to $s$, the algorithm observes that they all connect to vertices in $V_1$ and have the same weight. It therefore concludes that the edge weights within $V_1$ are also irrelevant to the $s$–$t$ distance. Consequently, the algorithm only needs to inspect the edge weights of the edges between $V_1$ and $V_2$. The overall query complexity is therefore $O(n)$.

On the other hand, for any algorithm, it is necessary to query all edges between $V_1$ and $V_2$ to find the one with the minimum weight. For an algorithm that aims to be instance-optimal, as explained before, we should regard the order of the edges in the adjacency lists as being chosen adversarially. 
Therefore, the expected time complexity is $\Omega(n^{2/3}\Delta^{2/3})$ since vertices in $V_1$ have degrees $\Theta(n^{1/3}\Delta^{1/3})$ and vertices in $V_2$ have degrees $\Theta(\Delta)$. 
    Therefore, any algorithm must spend $\Omega(n^{2/3}\Delta^{2/3})$ time in the worst input.
\end{proof}

By a small modification of the above proof, we get~\Cref{thm:undirected-fixed-lower-m}.

\begin{proof}[Proof of~\Cref{thm:undirected-fixed-lower-m}]
    The proof follows from the proof of~\Cref{thm:undirected-fixed-lower}.
    Note that the graph we constructed in the proof of~\Cref{thm:undirected-fixed-lower} has $\Theta(n^{2/3}\Delta^{2/3}+n)$ edges. 
    We now have a given $m\in[n,n^2]$ for the number of edges in the graph, and we set $|V_1|=\Theta(\sqrt{m}),|V_2|=\Theta(n/\sqrt{m}),|V_3|=\Theta(m^{1.5}/n)$ for $m\le n^{4/3}$.  
 \end{proof}

%% file: analysis_directed.tex
\section{More Instance-Optimality}\label{sec:more-positive}

In this section, we present upper bounds on the instance-optimality ratio of bidirectional Dijkstra, in all settings other than that of order-oblivious undirected graphs, which we have already covered in~\Cref{sec:order-oblivious-undirected}.

Comparing these upper bounds with the lower bounds established in~\Cref{sec:order-dependent-directed,sec:order-oblivious-directed,sec:order-dependent-undirected}, we note that our analysis of bidirectional Dijkstra is tight in all settings.
More surprisingly, bidirectional Dijkstra achieves the asymptotically best possible instance-optimality ratio among all algorithms in every setting, 
except that of order-oblivious directed graphs, where our results are inconclusive.

Additionally, we show strong upper bounds on the instance-optimality ratio of bidirectional Dijkstra on easy instances, i.e.\ instances with low instance-complexity.

\subsection{Directed Graphs}\label{subsec:upper-directed}
We begin by presenting our upper-bound results for directed graphs.
These results are stated in the order-dependent model and, as discussed in~\Cref{sec:model}, also hold in the order-oblivious model.

When parameterizing in $n$ and $m$, we get the following upper bound.

\begin{theorem}
\label{thm:directed-fixed-d}
    In the order-dependent model, the instance-optimality ratio of \Cref{alg:bi-dj-new} is $O(\lceil m/n \rceil)$ on directed graphs.
\end{theorem}

If we instead parameterize in $n$ and $\Delta$, we get the following upper bound.

\begin{theorem}
\label{thm:directed-fixed}
    In the order-dependent model, the instance-optimality ratio of \Cref{alg:bi-dj-new} is $O(\lceil\Delta^2/n\rceil)$ on directed graphs.
\end{theorem}

Finally, if we restrict to easy instances, we get the following upper bound.

\begin{theorem}
\label{thm:directed-tau}
    In the order-dependent model, the instance-optimality ratio of  \Cref{alg:bi-dj-new} is $O(h)$ on directed graphs when there are at most $h\ge 1$ vertices with degree $\ge n-10\tau$.
\end{theorem}

In the following, we will first prove~\Cref{thm:directed-fixed}, building on which we prove~\Cref{thm:directed-fixed-d} and~\Cref{thm:directed-tau}.  

\subsubsection{Proof of~\Cref{thm:directed-fixed}}
\begin{proof}[Proof of~\Cref{thm:directed-fixed}]
The proof builds on several lemmas. First, we claim

\begin{lemma}
\label{lem:Et-Delta^2}
    If $Q_A<|E_t|/10$, then $|E_t|=O(\Delta^2)$.
\end{lemma}

 \begin{proof}
We assume for contradiction that $|E_t|>10\Delta^2$. We will then be able to apply either Construction~(\hyperref[cons:a]{a}) or Construction~(\hyperref[cons:b]{b}) from~\Cref{sec:framework} to construct a graph $G’$ such that $d_{G’}(s,t)<d_G(s,t)$, while $A$ produces the same output on both $G$ and $G’$ with probability at least $1/2$. 
This contradicts the assumption that $A$ is a good algorithm that outputs the correct answer on every instance with probability at least $9/10$. The details are given below.

Since $Q_A<|E_t|/10$, at least $\frac{3}{5}|E_t|$ edges in $E_t$ are $1/4$-bad. Otherwise, at least $\frac{2}{5}|E_t|$ edges in $E_t$ would be $1/4$-good, contributing more than $|E_t|/10$ to $Q_A$, a contradiction. Similarly, at least $\frac{3}{5}|E_s|$ edges in $E_s$ are $1/4$-bad. Fix one such edge $(u_1,v_1)\in E_s$. There are at most $\Delta^2$ edges $(u_2,v_2)\in E_t$ satisfying $(u_1,v_2)\in E$, and at most another $\Delta^2$ edges satisfying $(u_2,v_1)\in E$. We call these at most $2\Delta^2$ edges \emph{forbidden edges} for future reference.

Since $|E_t|>10\Delta^2$, it follows that $\frac{3}{5}|E_t|>6\Delta^2$. Therefore, at least $6\Delta^2$ edges in $E_t$ are $1/4$-bad. Excluding the at most $2\Delta^2$ forbidden edges, at least $4\Delta^2>0$ $1/4$-bad edges remain. Fix one such remaining edge $(u_2,v_2)\in E_t$.

If $u_1=u_2$ or $v_1=v_2$, then there is already a bad edge from $V_s$ to $V_t$. We therefore apply Construction~(\hyperref[cons:a]{a}) from~\Cref{sec:framework} to construct $G’$.
Otherwise, if $u_2=v_1$, then the two $1/4$-bad edges already form a path from $u_1$ to $v_2$. Again, we apply Construction~(\hyperref[cons:a]{a}) to construct $G’$.

In the remaining case, we apply Construction~(\hyperref[cons:b]{b}) by replacing $(u_1,v_1)$ and $(u_2,v_2)$ with $(u_1,v_2)$ and $(u_2,v_1)$. 
This completes the proof. 
\end{proof}

Furthermore, we claim

\begin{lemma}
\label{lem:QA-Delta^2}
    If $Q_A< |E_t|/(10\lceil\Delta^2/n\rceil)$, then $Q_A=\Omega(n)$.
\end{lemma}

\begin{proof}
We assume for contradiction that $Q_A<n/20$.  
We will then be able to apply Constructions~(\hyperref[cons:a]{a}) and~(\hyperref[cons:c]{c}) from~\Cref{sec:framework} to construct $G'$. We consider the two cases $\Delta<n/2$ and $\Delta\ge n/2$ separately, as described below.

\paragraph{Case $\Delta<n/2$.}
Since $Q_A<|E_t|/10$, there exist a $1/10$-bad edge $(u_1,v_1)\in E_s$ and a $1/10$-bad edge $(u_2,v_2)\in E_t$.
Since $Q_A<n/20$, there exists at least one $1/20$-bad vertex $w$. If $(u_1,w)$ already exists, we leave it unchanged. Otherwise we remove the edge $(u_1,v_1)$ and add the edge $(u_1,w)$. We do similar to $v_2$ to connect $(w,v_2)$. 

In particular, consider the vertex $v_1$. Since $\Delta<n/2$, there are at least $n/2$ vertices that are not in-neighbors of $v_1$. Among them, there exists a $1/10$-bad vertex $x$, since $Q_A<n/20$. We add the edge $(x,v_1)$ to restore the in-degree of $v_1$. By a symmetric argument, there exists a $1/10$-bad vertex $y$ that is not an out-neighbor of $u_2$, and we add the edge $(u_2,y)$ to restore the out-degree of $u_2$. 

As a result, the degrees of at most one $1/20$-bad vertex and two $1/10$-bad vertices are modified, and at most two $1/10$-bad edges are removed. 
But now we have a length-two path $u_1\rightarrow w\rightarrow v_2$ connecting $V_s$ and $V_t$, thereby completing Construction~(\hyperref[cons:c]{c}).

\paragraph{Case $\Delta \ge n/2$.}
    We have
    \begin{align*}
        Q_A< \frac{|E_t|}{10 \Delta^2/n} 
        \leq \frac{|V_t| \Delta}{10 \Delta^2/n} 
        \leq \frac{|V_t| n}{10 \Delta}
        \leq \frac{|V_t|}{5}.
    \end{align*}
    Similarly, $Q_A< |V_s|/5$.
    Hence, there exist a $1/5$-bad vertex $u\in V_s$ and a $1/5$-bad vertex $v\in V_t$. 
If $(u,v)$ exists, we reduce its edge weight so that $d_{G’}(s,t)<d_G(s,t)$. Otherwise, we add $(u,v)$ and set its edge weight so that $d_{G’}(s,t)<d_G(s,t)$.
\end{proof}

Combining \cref{lem:Et-Delta^2,lem:QA-Delta^2} gives \Cref{thm:directed-fixed}. 
\end{proof}

\subsubsection{Proof of~\Cref{thm:directed-fixed-d}}
\begin{proof}[Proof of~\Cref{thm:directed-fixed-d}]
The proof builds on the following lemma, which is analogous to \cref{lem:QA-Delta^2}, but we now consider $m$ instead of $\Delta$.

\begin{lemma}
\label{lem:QA-d}
    If $Q_A< |E_t|/(40\lceil m/n \rceil)$, then $Q_A=\Omega(n)$.
\end{lemma}

\begin{proof}
    Assume for contradiction that $Q_A<n/20$. Let $B$ denote the set of $1/10$-bad vertices, then $|B|>n/2$. Let $w$ be an arbitrary vertex in $B$. We are going to connect $w$ to (some vertex in) $V_t$. Since $Q_A<|E_t|/40$, at least half of the edges in $E_t$ are $1/10$-bad. These edges come from at least $|E_t|/(2|V_t|)$ different vertices, as each vertex can have at most $|V_t|$ edges to $V_t$.
    \begin{itemize}
        \item If $|E_t|/(2|V_t|) > 2m/n$, then at least one of these vertices has an out-degree of at most $n/2$. So there exists a $1/10$-bad edge $(u_2,v_2)\in E_t$ such that the out-degree of $u_2$ is no larger than $n/2$. If $v_2=w$ or $(w,v_2)\in E$, we are done. Otherwise we replace $(u_2,v_2)$ with $(w,v_2)$, and then fix the out-degree of $u_2$. Since $\dout(u_2) \le n/2<|B|$, we can add an edge $(u_2,y)$ for some $y\in B$.
        \item If $|E_t|/(2|V_t|) \le 2m/n$, then $Q_A<|E_t|/(40m/n)<|V_t|/10$. So $V_t \cap B\neq\emptyset$. Let $v_2\in V_t\cap B$. If $v_2=w$ or $(w,v_2)\in E$, we are done. Otherwise we add the edge $(w,v_2)$.
    \end{itemize}
    In both cases, we connect $w$ to some $v_2\in V_t$. Symmetrically, we connect some $u_1\in V_s$ to $w$. Then we connect $u_1$ to $v_2$ through $w$. We then apply Construction~(\hyperref[cons:a]{a}) to construct $G'$. 
\end{proof}

Combining \cref{lem:QA-d} and the fact that $|E_t|=O(m)$ gives~\Cref{thm:directed-fixed-d}. 
\end{proof}

\subsubsection{Proof of~\Cref{thm:directed-tau}}

\begin{proof}[Proof of~\Cref{thm:directed-tau}]
Recall that in \Cref{thm:directed-tau}, we are going to prove that in the order-dependent model, the instance-optimality ratio of \Cref{alg:bi-dj-new} is $O(h)$ on directed graphs when there are at most $h\ge 1$ vertices with degree at least $n-10\tau$. Here we recall that $\tau$ denotes the instance-complexity. We are going to prove that $|E_t|=O(h Q_A)$, which in turn implies the desired $O(h)$ instance-optimality ratio of~\Cref{alg:bi-dj-new}.

    Let $B$ denote the set of $1/10$-bad vertices in $G$. Then every vertex in $V\setminus B$ is visited by $A$ with probability at least $1/10$. Therefore, the instance-complexity satisfies $\tau\ge (n-|B|)/10$, or equivalently, $|B|\ge n-10\tau$. 

    Let $H$ denote the set of vertices whose in- or out-degree is at least $n-10\tau$, then $|H|=h$.
    Let $E_s'=\{(u,v)\in E_s \mid v\notin H\}$. Since every vertex in $H$ has at most $|V_s|$ in-edges from vertex in $V_s$, we have $|E_s'|\ge |E_t|-|V_t|h$. Similarly, we define $E_t'=\{(u,v)\in E_t \mid u\notin H\}$ and have $|E_t'|\ge |E_t|-|V_t|h$. 

    \begin{lemma}
    \label{clm:directed-tau}
        Either all edges in $E_s'$ are $1/10$-good, or all edges in $E_t'$ are $1/10$-good.
    \end{lemma}

    \begin{proof}
        Assume for contradiction that there are $1/10$-bad edges $(u_1,v_1)\in E_s'$ and $(u_2,v_2)\in E_t'$.
        Let $w$ be an arbitrary vertex in $B$.
        We are going to connect $u_1$ to $w$. If $u_1=w$ or $(u_1,w)\in E$, we are done. Otherwise we replace $(u_1,v_1)$ with $(u_1,w)$, and then fix the in-degree of $v_1$.
        Since $v_1\notin H$, the indegree of $v_1$ is thus less than $n-10\tau\le |B|$, so we can add an edge $(x,v_1)$ for some $x\in B$.
        Symmetrically, we connect $w$ to $v_2$. Then we are able to apply Construction~(\hyperref[cons:a]{a}) to construct $G'$. 
    \end{proof}

    By \cref{clm:directed-tau}, $Q_A=\Omega(|E_t|-|V_t|h)$. On the other hand, $Q_A=\Omega(|V_t|)$ since otherwise there are two bad vertices in $V_s$ and $V_t$ respectively. So $|E_t|=(|E_t|-|V_t|h)+|V_t|h=O(hQ_A)$.
\end{proof}

\subsection{Undirected Graphs} \label{subsec:upper-undirected}
We now present our upper-bound results for undirected graphs.
Again, these are stated in the order-dependent model, but also hold in the order-oblivious model by definition.

We first get the following upper bound when parameterizing in $n$ and $m$. 

\begin{theorem}\label{thm:undirected-fixed-m}
    In the order-dependent model, the instance-optimality ratio of \Cref{alg:bi-dj-new} is $O(\min\{\lceil m/n \rceil,n^{1/3}\})$ on undirected graphs.
\end{theorem}

If we instead parameterize in $n$ and $\Delta$, we get the following upper bound.

\begin{theorem}\label{thm:undirected-fixed}
    In the order-dependent model, the instance-optimality ratio of \Cref{alg:bi-dj-new} is $O(\lceil\Delta^{2/3}/n^{1/3}\rceil)$ on undirected graphs.
\end{theorem}

Finally, if we restrict to instances with instance-complexity $\tau<n/10$, then we have

\begin{theorem}
\label{thm:undirected-tau}
    In the order-dependent model, \Cref{alg:bi-dj-new} is instance-optimal on undirected graphs with $\tau < n/10$.
\end{theorem}

In the following, we will first prove~\Cref{thm:undirected-tau}, building on which we prove~\Cref{thm:undirected-fixed-m} and~\Cref{thm:undirected-fixed}. 

\subsubsection{Proof of \Cref{thm:undirected-tau}}
\begin{proof}[Proof of \Cref{thm:undirected-tau}]
The proof builds on the following lemma. 
\begin{lemma}
\label{lem:undirected-QA}
    If $Q_A<|E_t|/10$, then $Q_A\ge n/10$.
\end{lemma}

\begin{proof}
    We assume for contradiction that $Q_A<n/10$. Then there must exist a $1/10$-bad vertex $w$.
    Since $Q_A<|E_t|/10$, there must exist $1/10$-bad edges $(u_1,v_1)\in E_s$ and $(u_2,v_2)\in E_t$.
    
    If $u_1,v_1,u_2,v_2$ are not all distinct, then $u_1$ and $v_2$ are already connected by bad edges. Then we are able to apply Construction~(\hyperref[cons:a]{a}) to construct $G'$. 
    
    Otherwise, we are going to connect $u_1$ with $w$ (and symmetrically connect $v_2$ with $w$).
    If $(u_1,w)\in E$ or $(v_1,w)\in E$, then $u_1$ is already connected to $w$. Otherwise, we can replace $(u_1,v_1)$ by $\{(u_1,w),(v_1,w)\}$, without changing the degrees of $u_1$ and $v_1$. Then we are able to apply Construction~(\hyperref[cons:b]{b}) to construct $G'$. 
\end{proof}

\cref{lem:undirected-QA} directly implies the \Cref{thm:undirected-tau}.
\end{proof}

\subsubsection{Proof of \Cref{thm:undirected-fixed-m} and \Cref{thm:undirected-fixed}}
The proofs build on the following lemma.
\begin{lemma}
\label{lem:undirected-Et}
    If $Q_A<|E_t|/40$, then $|E_t|/Q_A = O(\Delta/\sqrt{|E_t|})$.
\end{lemma}

\begin{proof}
    For a $1/10$-bad edge $(u,v)\in E_s$, we say it is \emph{very bad} if (1) there is another $1/10$-bad edge $(u',v)\in E_s$ such that $d(s,u')\le d(s,u)$, or (2) $v$ is a $1/10$-bad vertex. We symmetrically define very bad edges in $E_t$.
    Since $Q_A<|E_t|/40$, at least half of the edges in $E_t$ are $1/10$-bad, and there are at most $|E_t|/4$ $1/10$-good vertices.
    For each of these $1/10$-good vertices, it has at most one incident edge that is a $1/10$-bad edge in $E_t$ but not very bad.
    Therefore, at least $1/8$ of the edges in $E_t$ are very bad.
    Similarly, at least $1/8$ of the edges in $E_s$ are very bad.
    
    Let $S$ be the set of quadruples $(u_1,v_1,u_2,v_2)$ such that $(u_1,v_1)$ is a very bad edge in $E_s$ and $(u_2,v_2)$ is a very bad edge in $E_t$. Then
    \begin{equation}
    \label{eqn:S-lower}
        |S| = \Omega(|E_t|^2).
    \end{equation}

    The next step is to give an upper bound of $|S|$.

    \begin{claim}
    \label{clm:undirected-good-1}
        For any $(u_1,v_1,u_2,v_2)\in S$, at least one of $(u_1,v_2)$ and $(v_1,u_2)$ is a $1/10$-good edge.
    \end{claim}

    \begin{proof}
        If one of them is a $1/10$-bad edge, then $u_1$ is connected to $v_2$ by bad edges.
        If both of them do not exist, then we can replace $\{(u_1,v_1),(u_2,v_2)\}$ by $\{(u_1,v_2),(v_1,u_2)\}$.
    \end{proof}

    \begin{claim}
        For any $(u_1,v_1,u_2,v_2)\in S$, at least one of $(u_1,u_2)$ and $(v_1,v_2)$ is a $1/10$-good edge.
    \end{claim}

    \begin{proof}
        If one of them is a $1/10$-bad edge, then $u_1$ is connected to $v_2$ by bad edges.
        Now assume both of them do not exist. If both $v_1$ and $u_2$ are $1/10$-bad, we can replace $(u_1,v_1)$ by $(u_1,u_2)$.
        Otherwise, assume $v_1$ is $1/10$-good by symmetry. We replace $\{(u_1,v_1),(u_2,v_2)\}$ by $\{(u_1,u_2),(v_1,v_2)\}$.
        By the definition of very bad edges, we connect some $u_3\in V_s$ to $v_2$ where $(u_3,v_1)$ is a $1/10$-bad edge and $d(s,u_3)\le d(s,u_1)$. Since $d(s,u_3)\le d(s,u_1)$, the shortest path from $s$ to $u_3$ will not be affected by removing $(u_1,v_1)$.
    \end{proof}

    By the two claims, the quadruple contains one $1/10$-good edge from $\{(u_1,v_2),(v_1,u_2)\}$ and one $1/10$-good edge from $\{(u_1,u_2),(v_1,v_2)\}$.
    Any such two edges share an endpoint.
    Thus some vertex in $\{u_1,v_1,u_2,v_2\}$ is incident to two $1/10$-good edges within the quadruple.
    We charge the quadruple to this vertex.
    
    For a vertex $v$, let $k_v$ be the number of $1/10$-good edges incident to $v$.
    Suppose a quadruple is charged to $v$.
    Then $v$ is incident to one very bad edge of the quadruple and to two $1/10$-good edges inside the quadruple.
    There are at most $\Delta$ choices for the very bad edge and at most $k_v^2$ choices for the two $1/10$-good edges.
    Also, there are at most $|E_t|$ choices for the other very bad edge of the quadruple.
    So $v$ is charged at most $\Delta\min\{k_v^2,|E_t|\} \leq \Delta \sqrt{|E_t|} k_v$ times.
    
    We further note that $\sum_v k_v$ is twice the number of $1/10$-good edges.
    Since each $1/10$-good edge is queried with probability at least $1/10$, the number of good edges is at most $10Q_A$.
    Summing over all vertices, we get
    \begin{equation}
    \label{eqn:S-upper}
        |S| \leq \Delta \sqrt{|E_t|}\sum_{v}  k_v = O(Q_A\Delta \sqrt{|E_t|}).
    \end{equation}
    Combining \cref{eqn:S-lower,eqn:S-upper} proves the lemma.
\end{proof}

\begin{proof}[Proof of~\Cref{thm:undirected-fixed}]
    If $|E_t| \le n^{2/3}\Delta^{2/3}$, the theorem follows from \cref{lem:undirected-QA}; if $|E_t| > n^{2/3}\Delta^{2/3}$, the theorem follows from \cref{lem:undirected-Et}.
\end{proof}

\begin{proof}[Proof of~\Cref{thm:undirected-fixed-m}]
The instance-optimality ratio of \Cref{alg:bi-dj-new} is $O(\lceil m/n \rceil)$ by \cref{lem:undirected-QA} and $|E_t|=O(m)$, or $O(n^{1/3})$ by \cref{thm:undirected-fixed} and $\Delta=O(n)$.
\end{proof}

\input{fewhighdegree}

%% file: fewhighdegree.tex
% \christian{
\subsection{Few high-degree vertices}
We now extend our results on the instance-optimality ratio of bidirectional Dijkstra on easy instances by showing that it admits an upper bound of $O(\hD+\lceil D^2/n\rceil)$ on directed graphs with at most $\hD$ vertices whose in-degree or out-degree exceeds $D$.

\begin{theorem}\label{thm:directed-fixed-h-D}
In the order-dependent model, the instance-optimality ratio of~\Cref{alg:bi-dj-new} is $O(\hD+\lceil D^2/n\rceil)$ on directed graphs with at most $\hD$ vertices whose in-degree or out-degree is larger than $D$.
\end{theorem}

Let $\HD$ be the set of vertices of degree larger than $D$, and write $\hD = |\HD|$.

\begin{lemma}
\label{lem:Et-h-D}
    If $Q_A<|E_t|/10$, then $|E_t|=O( n \hD+D^2)$.
\end{lemma}
\begin{proof}
    Assume for contradiction that $|E_t|>10(n\hD+D^2)$.
    At least $\frac{3}{5}|E_t|>6(n\hD+D^2)$ edges in $E_t$ are $1/4$-bad; otherwise these $1/4$-good edges in $E_t$ already contribute more than $|E_t|/10$ to $Q_A$.
    Recalling that $|E_s| \ge |E_t|$, we thus have that at least $\frac{3}{5}|E_s|$ edges in $E_s$ are $1/4$-bad.
    Since $|\HD| = \hD$, at most $2n\hD$ edges have an endpoint in $\HD$.
    Avoiding such edges, there remain more than $6(n\hD+D^2)-2n\hD  \geq 1$ edges in $E_s$ which are $1/4$-bad.
    In particular, there is a $1/4$-bad edge $(u_1,v_1) \in E_s$ with $u_1 \not\in \HD$ and $v_1 \not\in \HD$.

    We now find a $\frac14$-bad edge $(u_2,v_2)\in E_t$ such that $(u_1,v_2)\notin E$ and $(u_2,v_1)\notin E$.
    There are at most $2(n\hD+D^2)$ edges $(u_2,v_2) \in E_t$ that do not satisfy this: at most $n\hD+D^2$ with $(u_1,v_2) \in E$, since $u_1$ has at most $D$ neighbors, and at most $\hD$ of these have degree larger than $D$, and analogously, at most $n\hD+D^2$ with $(u_2,v_1) \in E$.
    Thus, at least $6(n\hD+D^2)-2(n\hD+D^2) \ge 1$ edges in $E_t$ are $1/4$-bad and satisfy both conditions, so an edge $(u_2,v_2)$ like we were looking for indeed exists.
    Note that $u_1\neq u_2$, $v_1\neq v_2$, and $u_1\neq v_2$, otherwise we can apply Construction~(\hyperref[cons:a]{a}) to construct $G'$. 

    If $u_2=v_1$, then the two $1/4$-bad edges already constitute a path from $u_1$ to $v_2$. 
    Otherwise we replace $(u_1,v_1)$ and $(u_2,v_2)$ by $(u_1,v_2)$ and $(u_2,v_1)$.
\end{proof}

\begin{lemma}
\label{lem:QA-h-D}
    If $Q_A< |E_t|/(10\lceil \hD +D^2/n\rceil)$, then $Q_A=\Omega(n)$.
\end{lemma}
\begin{proof}
    Assume, for contradiction, that $Q_A<n/20$. Then more than $n/2$ vertices are $1/10$-bad.

    \emph{Case 1.} Assume $D\ge n/2$.
    Since $\lceil \hD +D^2/n\rceil\ge n/4$, we have
    \[
        Q_A<\frac{|E_t|}{10\lceil \hD +D^2/n\rceil}
        \leq \frac{|V_t|n}{10\lceil \hD +D^2/n\rceil}
        \leq \frac{2}{5}|V_t|.
    \]
    Similarly, using $|E_t|\le |E_s|$, we get $Q_A<\frac25|V_s|$. 
    So there are $2/5$-bad vertices $u\in V_s$ and $v\in V_t$.
    If $(u,v)\in E$, this gives a bad path from $V_s$ to $V_t$.
    Otherwise we add the edge $(u,v)$, again giving a bad path. We can then apply Construction~(\hyperref[cons:a]{a}) to construct $G'$, contradicting the assumption that $A$ is a good algorithm.  
    % By~\Cref{bad path lemma that shuyi was talking bout}, we have a contradiction.

    \emph{Case 2.} Assume the opposite, namely $D<n/2$.
    Since $Q_A < |E_t|/10 \leq |E_s|/10$, there are $1/10$-bad edges in both $E_s$ and $E_t$.
    Let $w$ be a $1/10$-bad vertex.
    We will now create a bad path from $V_s$ to $V_t$ through $w$, starting with creating a path from $V_s$ to $w$.

    \emph{Case 2.1.} Assume that there exists a $1/10$-bad edge $(u_1,v_1) \in E_s$ with $v_1 \not\in \HD$.
    If $u_1=w$ or $(u_1,w)\in E$, we have a bad path from $V_s$ to $w$.
    Otherwise, replace $(u_1,v_1)$ by $(u_1,w)$.
    If $v_1$ is $1/10$-good, we need to fix its in-degree.
    Since $v_1 \not\in \HD$, the in-degree of $v_1$ is at most $D<n/2$, so there is a $1/10$-bad vertex $x\neq v_1$ with $(x,v_1)\not\in E$.
    We can add $(x,v_1)$, fixing the in-degree of $v_1$.
    We have now created a bad path from $V_s$ to $w$ and changed the degrees of only a few $1/10$-bad vertices.

    \emph{Case 2.2.} Assume the opposite, namely that every edge $(u_1,v_1) \in E_s$ with $v_1 \notin H$ is $1/10$-good.
    There are fewer than $|E_s|/\lceil \hD +D^2/n\rceil$ such edges.
    The remaining edges $(u_1,v_1) \in E_s$ lie in $V_s \times \HD$, so there are at most $|V_s| \hD $ of them.
    In total,
    \[
        |E_s|<\frac{|E_s|}{\lceil \hD +D^2/n\rceil}+h|V_s|.
    \]
    If $\lceil \hD +D^2/n\rceil\ge 2$, this implies $|E_s|<2 \hD  |V_s|\le 2\lceil \hD +D^2/n\rceil |V_s|$, and hence
    \[
        Q_A<\frac{|E_t|}{10\lceil \hD +D^2/n\rceil}
        \le \frac{|E_s|}{10\lceil \hD +D^2/n\rceil}
        <\frac{|V_s|}{5}.
    \]
    If $\lceil \hD +D^2/n\rceil=1$, then $\hD =0$ is impossible in the current case, because then $\HD=\emptyset$ and the case assumption would imply that all edges of $E_s$ are $1/10$-good, contradicting $Q_A<|E_s|/10$. So instead, $\hD =1$ and $D=0$, in which case $|E_s|\le |V_s|$, giving again $Q_A<|V_s|/5$.
    This means that there is a $1/5$-bad vertex $u$ in $V_s$.
    If $u=w$ or $(u,w) \in E$, we already have a bad path from $u$ to $w$.
    Otherwise we add an edge $(u,w)$.
    We have thus again created a bad path from $V_s$ to $w$ and changed the degrees of only the $1/5$-bad vertices $u$ and $w$.

    Together, Case 2.1 and 2.2 tell us that we can connect $V_s$ to $w$ by a bad path.
    Symmetrically, we can connect $w$ to $V_t$ by a bad path, giving the wished-for bad path from $V_s$ to $V_t$.
    By a union bound, $A$ distinguishes the modified graph from the original with probability less than $1/2$, which contradicts the assumption that $A$ is a good algorithm that outputs the correct answer on every instance with probability at least $9/10$. 
    \end{proof}

Combining~\cref{lem:Et-h-D,lem:QA-h-D}, we get~\cref{thm:directed-fixed-h-D}.

%%% Local Variables:
%%% mode: LaTeX
%%% TeX-master: "main"
%%% End:

%% file: no_degree.tex
\section{No degree queries}
\label{sec:nodegree}

In this section, we show that the instance-non-optimality of bidirectional Dijkstra on directed graphs is mainly caused by degree queries (i.e., $\indeg$ and $\outdeg$ for directed graphs, and $\deg$ for undirected graphs). Specifically, we now assume that the only operation available to an algorithm is to take a previously seen vertex and request its next in- or out-neighbor, together with the weight of the corresponding edge. We refer to this as the \emph{no-degree} model.

We retain the distinction between the order-dependent and order-oblivious models: in the former, the ordering is considered part of the input, whereas in the latter, it is not. Consequently, in the order-dependent model, an instance-smart algorithm may exploit the ordering by stopping its requests for subsequent in- or out-neighbors when the most important edges appear near the beginning of an adjacency list. In contrast, in the order-oblivious model, an instance-smart algorithm can only treat the order in which incident edges are returned as uniformly random.

In the no-degree model, we show that bidirectional Dijkstra is instance-optimal in almost all settings, regardless of whether the graph is directed or undirected and whether the model is order-dependent or order-oblivious. The only exception is the order-dependent setting on directed graphs, where bidirectional Dijkstra remains instance-non-optimal.

\subsection{Order-Oblivious Instance-Optimality for Directed Graphs}

We now use the following construction of $G'$, which does not preserve degrees.
\begin{theorem}\label{thm:oblivious-directed-nodegree}
    In the order-oblivious no-degree model, \Cref{alg:bi-dj-new} is instance-optimal on directed graphs.
\end{theorem}

\begin{proof}
Assume for contradiction that $Q_A<|E_t|/20$ edges, where we recall from~\Cref{eqn:querynumber} that $Q_{A}$ denotes the expected number of vertices and edges visited by $A$. We also recall from~\Cref{subsec:property_bj} that we use $E_s$ and $E_t$ to denote the sets of edges visited by the forward and backward searches, respectively, in the bidirectional Dijkstra's algorithm shown in~\Cref{alg:bi-dj-new}. By~\cref{lem:time_bj}, $|E_t|\le |E_s|\le |E_t|+1$, and the number of queries performed by bidirectional Dijkstra is $\Theta(|E_s|)=\Theta(|E_t|)$.
Then there exists a $1/20$-bad edge $(u_1,v_1)\in E_s$ and a $1/20$-bad edge $(u_2,v_2)\in E_t$.

If $(u_1,v_2)$ exists in $G$, then, by~\Cref{lem:oblivious-edge-bad}, $(u_1,v_2)$ is $1/5$-bad, since $u_1$ has a $1/20$-bad outgoing edge $(u_1,v_1)$ and $v_2$ has a $1/20$-bad incoming edge $(u_2,v_2)$. We can then apply Construction~(\hyperref[cons:a]{a}) to construct $G’$. 
Specifically, we construct $G’$ by changing the edge weight $\ell(u_1,v_2)$ so that $\ell_{G'}(u_1,v_2)<d_G(s,t)-d_G(s,u_1)-d_G(v_2,t)$. This is possible by~\cref{lem:dist-st}, which shows that $d_G(s,u_1)+d_G(v_2,t)<d_G(s,t)$ for $u_1\in V_s$ and $v_2\in V_t$. Consequently, $d_{G’}(s,t)<d_G(s,t)$.
Since the edge weight $\ell(u_1,v_2)$ is the only difference between $G$ and $G’$, and $(u_1,v_2)$ is $1/5$-bad, i.e., it is visited by $A$ with probability at most $1/5$, $A$ can distinguish $G$ from $G'$ with probability at most $1/5$. This contradicts the assumption that $A$ is a good algorithm that outputs the correct answer on every instance with probability at least $9/10$.

Now suppose that $(u_1,v_2)$ does not exist in $G$. We construct $G’$ by adding $(u_1,v_2)$ and setting $\ell(u_1,v_2)<d_G(s,t)-d_G(s,u_1)-d_G(v_2,t)$. Specifically, we insert $(u_1,v_2)$ at a uniformly random position in the out-adjacency list of $u_1$. Analogously, we insert $(u_1,v_2)$ at a uniformly random position in the in-adjacency list of $v_2$. Recall that, in the order-oblivious model, the order in which the graph oracle returns edges in an adjacency list is considered uniformly random.
In particular, $(u_1,v_2)$ appears after $(u_1,v_1)$ in the out-adjacency list of $u_1$ with probability $1/2$. Since $(u_1,v_1)$ is $1/20$-bad, it is visited by $A$ with probability at most $1/20$. Analogously, $(u_2,v_2)$ is visited by $A$ with probability at most $1/20$. Thus, the probability that either of the two edges is visited is at most $1/10$. So that all three edges remain unvisited with probability at least $9/40 \ge 1/5$.

Since degree queries are not available, on this event, $A$ has not visited all edges in the out-adjacency list of $u_1$ or the in-adjacency list of $v_2$, and therefore cannot detect the change in the out-degree of $u_1$ and in-degree of $v_2$. As a result, $A$ cannot distinguish $G$ from $G’$ with probability at least $1/5$, contradicting the assumption that $A$ is a good algorithm. 
\end{proof}

\subsection{Order-Dependent Instance-Non-Optimality for Directed Graphs}
We now show that disallowing degree queries does not eliminate instance-non-optimality in the order-dependent model on directed graphs.
\begin{theorem}\label{thm:directed-dependent-no-degree-lower}
    In the order-dependent no-degree model, the instance-optimality ratio of every algorithm is $\Omega(\lceil\Delta/\sqrt n\rceil)$ on directed graphs.
\end{theorem}
\begin{proof}
    The claim is trivial if $\Delta < \sqrt n$, so assume the opposite.
    Take the construction of~\Cref{thm:directed-fixed-lower} but give every vertex in $V_3$ an edge to every vertex in $V_4$.
    Set $|V_1|=|V_2|=\Theta(\Delta)$, $|V_3|=|V_4|=\Theta(\sqrt n)$, and $|V_5|=n-\Theta(\Delta+\sqrt n)$, so that the number of vertices is $n$ and the maximum in- and out-degree is $\Delta$. In particular, we place the edges between $V_3$ and $V_4$ at the beginning of the out-adjacency list of every vertex in $V_3$ and the in-adjacency list of every vertex in $V_4$.

    An instance-smart algorithm can first visit all edges incident to $s$.
    From this, it can identify the set $V_1$ as the set of neighbors of $s$ with no out-edges, and identify $V_3$ as the others.
    Similarly, it can visit all in-neighbors of $t$, to identify the set $V_2$ as those with no in-edges, and identify $V_4$ as the others.
    It can also identify $V_5$ as an induced path by walking along out-edges starting from $t$.
    It has now counted all $n$ vertices, and certified that no path from $s$ to $t$ visits $V_1$ and $V_2$.
    
    Note that we now consider the order-dependent model, in which the ordering is part of the input graph. In particular, the edges between $V_3$ and $V_4$ are placed at the beginning of the out-adjacency list of every vertex in $V_3$. Hence, an instance-smart algorithm can inspect only the edges from $V_3$ to $V_4$, without visiting $V_1$ or $V_2$, by terminating the neighbor queries as soon as they start returning vertices in $V_1$. As a result, the instance-smart algorithm requires only $O(|V_3||V_4|)=O(n)$ queries to inspect all edges from $V_3$ to $V_4$. Since every $s$–$t$ path must go from $s$ to $V_3$, then to $V_4$, and finally to $t$, the algorithm can then determine the shortest-path distance.

    In contrast, consider any good algorithm. The algorithm must query every edge from $V_3$ to $V_4$ to ensure the correctly calculate the shortest $s$-$t$ path. Recall that an algorithm that aims to be instance-optimal must treat the ordering of the incident edges it receives as uniformly random, rather than relying on an instance-specific ordering. Consequently, on this graph, such an algorithm requires $\Omega(\sqrt{n}\Delta)$ queries. Comparing this with the $O(n)$ queries required by the instance-smart algorithm yields an instance-optimality ratio of $\Omega(\sqrt n\Delta/n)=\Omega(\Delta/\sqrt n)$.
    \end{proof}

\subsection{Order-Dependent Instance-Optimality for Undirected Graphs}

\begin{theorem}\label{thm:order-dependent-undirected-nodegree}
    In the order-dependent no-degree model,~\Cref{alg:bi-dj-new} is instance-optimal on undirected graphs.
\end{theorem}

\begin{proof}
Let $U_s\subseteq V_s$ and $U_t\subseteq V_t$ denote the sets of vertices in $V_s$ and $V_t$, respectively, that are incident to a $1/20$-bad edge. Assume without loss of generality that $|U_s|\le |U_t|$.

Assume, for contradiction, that $Q_A<|E_t|/20$, where we recall from~\Cref{eqn:querynumber} that $Q_{A}$ denotes the expected number of vertices and edges visited by $A$. Then there exists a $1/20$-bad edge $(u_1,v_1)\in E_s$ and a $1/20$-bad edge $(u_2,v_2)\in E_t$.

If $v_1=v_2$, or if $(v_1,v_2)$ is a $1/20$-bad edge in $G$, then we can simply apply Construction~(\hyperref[cons:a]{a}) to construct $G’$.

If $v_1\neq v_2$ and $(v_1,v_2)$ does not exist in $G$, we construct $G’$ by adding $(v_1,v_2)$ and setting $\ell_{G'}(v_1,v_2)<d_G(s,t)-d_G(s,u_1)-d_G(v_2,t)$. This is possible by~\cref{lem:dist-st}, which shows that $d_G(s,u_1)+d_G(v_2,t)<d_G(s,t)$ for $u_1\in V_s$ and $v_2\in V_t$. Consequently, $d_{G’}(s,t)<d_G(s,t)$.

Recall that we are now considering the order-dependent model, in which the ordering of edges within each adjacency list is part of the input. We place the newly added edge $(v_1,v_2)$ at the ``bottom'' of the adjacency lists of $v_1$ and $v_2$, in the sense that the probability that $A$ reaches $(v_1,v_2)$ when scanning either adjacency list is no larger than the probability of reaching any existing edge in that list. Since $(u_1,v_1)$ and $(u_2,v_2)$ are both $1/20$-bad, i.e., each is visited by $A$ with probability at most $1/20$, the new edge $(v_1,v_2)$ is visited through the adjacency list of either $v_1$ or $v_2$ with probability at most $1/20$, and hence is visited with probability at most $1/10$ in total.

Moreover, degree queries are not available. Therefore, unless $A$ scans an entire adjacency list, it cannot detect the change in the degree of $v_1$ or $v_2$ caused by adding $(v_1,v_2)$. It follows that $A$ can distinguish $G$ from $G’$ only if it visits $(u_1,v_1)$, $(u_2,v_2)$, or the newly added edge $(v_1,v_2)$. By a union bound, this occurs with probability at most $1/20+1/20+1/10=1/5$. Thus, $A$ cannot distinguish $G$ from $G’$ with sufficiently high probability, contradicting the assumption that $A$ is a good algorithm.

It is worth noting that this is precisely where the argument fails to extend to the order-dependent no-degree model on directed graphs. In the directed setting, the bad edge $(u_1,v_1)$ belongs to the in-adjacency list of $v_1$, whereas the new edge $(v_1,v_2)$ must be inserted into the out-adjacency list of $v_1$. For example, $v_1$ may have out-degree zero, in which case there is no bad edge in its out-adjacency list that can be used to hide the newly added edge $(v_1,v_2)$. Thus, even without degree queries, adding $(v_1,v_2)$ may be detectable by $A$.

It therefore remains only to consider the case where $(v_1,v_2)$ exists in $G$ and is $1/20$-good. In fact, this implies that, for every $1/20$-bad edge ${u_1,v_1}$ with $u_1\in U_s$ and every $v_2\in U_t$, there exists a $1/20$-good edge ${v_1,v_2}$. Hence, $v_1$ is incident to at least $|U_t|\ge |U_s|$ $1/20$-good edges, which is at least the number of $1/20$-bad edges between $v_1$ and $U_s$.

Let $g$ denote the number of $1/20$-good edges in $E$, and let $g_s$ and $b_s$ denote the numbers of distinct $1/20$-good and $1/20$-bad undirected edges in $E_s$, respectively. Summing over all such vertices $v_1$, we obtain $2g\ge b_s$, where the factor $2$ accounts for the fact that each $1/20$-good edge may be counted from both endpoints. Clearly, $g\ge g_s$. Moreover, each undirected edge occurs at most twice in $E_s$. Therefore, $|E_s|\le 2(g_s+b_s)\le 6g$.

Since every $1/20$-good edge is visited by $A$ with probability at least $1/20$, $A$ queries at least $(1/20)g\ge |E_s|/120$ edges in expectation, contradicting the assumption that $Q_A<|E_s|/20$.\end{proof}

\section{Zero-Weight Graphs and Unit-Weight Graphs}\label{sec:zero-unit-weight}
Throughout the paper, we consider only graphs with positive edge weights. In this section, we briefly discuss the cases where zero-weight edges are allowed, or the input graph is promised to have a smallest valid positive edge weight (e.g., an unweighted graphs). 

We note that~\cite{sosa_bidirectional_HaeuplerHRTT25} also establishes the instance-non-optimality of bidirectional Dijkstra on zero-weight and unweighted graphs, under the assumption that parallel edges are allowed. They obtain lower bounds of $\Omega(n/\log n)$ and $\Omega(\Delta)$ on the instance-optimality ratio for zero-weight and unweighted graphs, respectively. Both lower bounds are established in the order-dependent model, allowing an instance-smart algorithm to be told the correct path. To get an order-oblivious result, the footnote of~\cite[Section~5.1]{sosa_bidirectional_HaeuplerHRTT25} suggests coding the identifier of a vertex by giving it a neighbor who is a star whose degree is the identifier. Here we do not allow degree queries, and we only use constant degrees.

\begin{theorem}\label{thm:zero-weight-m}
When zero-weight edges are allowed, the instance-optimality ratio of every algorithm is $\Omega(\min\{n,m/\log n\})$ in the order-oblivious no-degree model, for both directed and undirected graphs.
\end{theorem}

\begin{proof}
Consider the graph $G$ illustrated in~\Cref{fig:undirected-zero-weight}. The graph contains two tree structures rooted at $s$ and $t$, respectively. We first describe the tree rooted at $s$.

\begin{figure}[!h]
        \centering
        \includegraphics[width=0.9\linewidth]{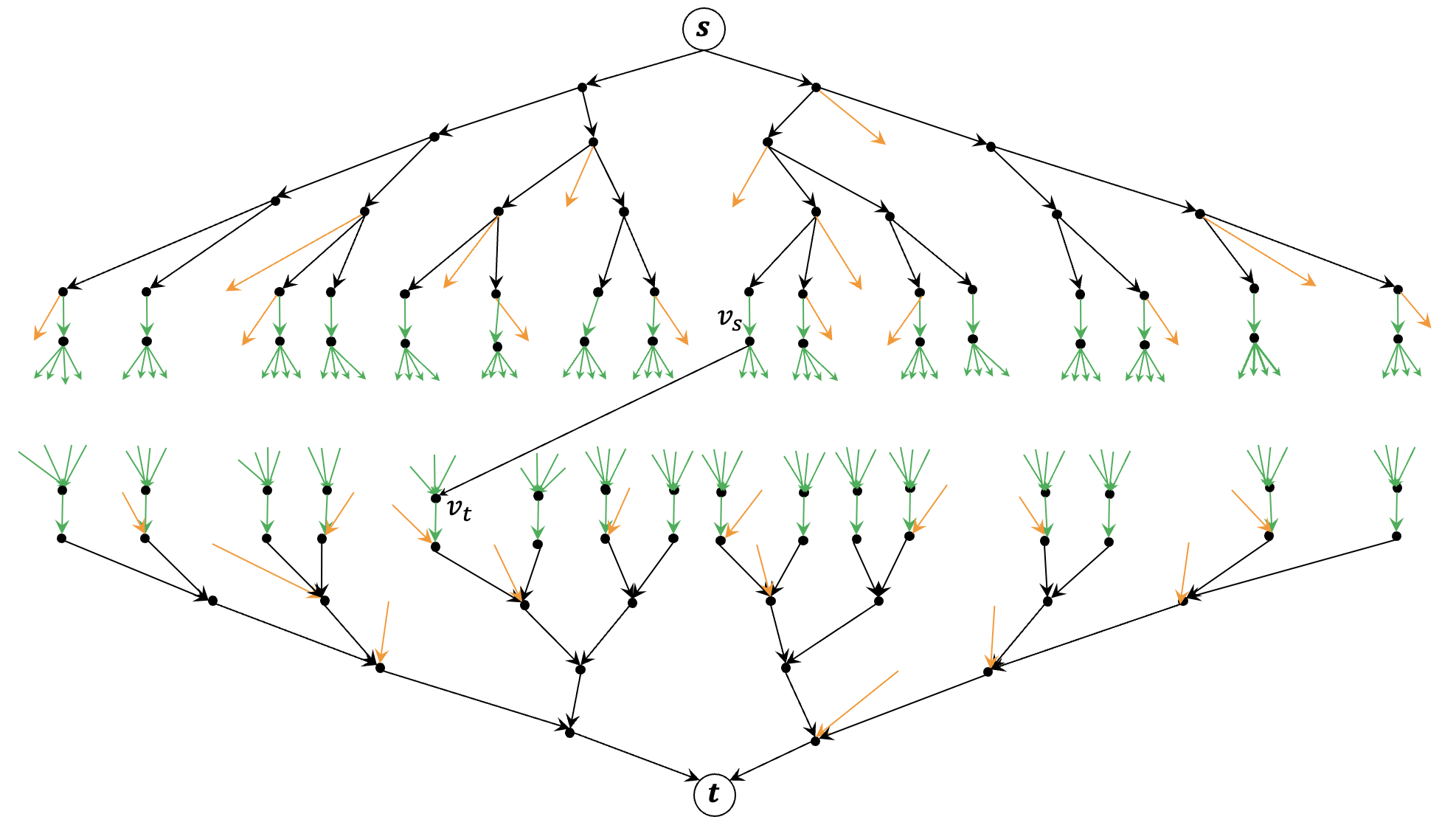}
        \vspace{-1mm}
        \caption{Lower-bound construction for graphs with zero or unit weights.}
        \label{fig:undirected-zero-weight}
    \end{figure}
Starting from a complete binary tree, we use vertex degrees to encode every root-to-leaf path. For each pair of sibling vertices, we attach an additional leaf to exactly one of them (see the orange edges in~\Cref{fig:undirected-zero-weight}), so that the two siblings have distinct degrees. This degree distinction indicates which of the two siblings lies on the encoded path. For each leaf of the original binary tree, we further add a new adjacent vertex and attach additional leaves to this new vertex so that it has out-degree $d$ (see the green edges in~\Cref{fig:undirected-zero-weight}). Consequently, a path from the root to any out-degree-$d$ vertex at the bottom level can be uniquely specified by the sequence of degrees encountered along the path.

Let $v_s$ denote an arbitrary out-degree-$d$ vertex at the bottom level of the tree rooted at $s$. 
The tree rooted at $t$ is constructed symmetrically: we choose an arbitrary in-degree-$d$ vertex $v_t$ at the bottom level. We assign weight $0$ to every edge in the graph. 
All other leaf vertices in the tree rooted at $s$ are disconnected from the tree rooted at $t$. Finally, we connect $v_s$ and $v_t$ by a zero-weight edge.

An instance-smart algorithm can ``know'' the degree sequence corresponding to the unique zero-length $s\to v_s\to v_t\to t$ path and therefore follow it directly. It thus requires only $O(\log n+d)$ queries to identify the zero-length $s$–$t$ path: $O(\log n)$ queries suffice to navigate from $s$ to $v_s$ and from $t$ to $v_t$ by following the degree-encoded paths, while an additional $O(d)$ queries suffice to identify the zero-weight edge between $v_s$ and $v_t$ by querying their neighbors.

In contrast, for an algorithm that aims to be instance-optimal over all instances, knowing the path code corresponding to one particular instance provides no useful information for other instances. Such an algorithm must therefore explore the whole graph to identify the zero-length path, resulting in $\Omega(m)$ query time. This yields an instance-optimality ratio of $\Omega(m/(\log{n}+d))=\Omega(\min \{ n, m/\log{n}\})$. 

The above construction is described in the context of directed graphs, but it extends straightforwardly to undirected graphs by replacing every directed edge with an undirected edge. The instance-smart algorithm can still navigate from $s$ to $v_s$ and from $t$ to $v_t$ by following the degree-encoded paths. The only difference is that, since the edges are now undirected, at each vertex the algorithm excludes the edge leading back to its parent and uses the degrees of the remaining neighbors to determine which vertex to visit next.
\end{proof}

\begin{theorem}
On unweighted graphs (or positively weighted graphs with a smallest valid positive weight), the instance-optimality ratio of every algorithm is $\Omega(\Delta)$ in the order-oblivious no-degree model, for both directed and undirected graphs.
\end{theorem}

\begin{proof}
Let $\ell_{\min}$ denote the smallest valid positive edge weight, or the unit weight, say $1$, if we restrict our attention to unweighted graphs. We consider the same graph as illustrated in~\Cref{fig:undirected-zero-weight}, but set the weight of every edge to $\ell_{\min}$ and replace each degree-$d$ vertex at the bottom level with a degree-$\Delta$ vertex. There is a unique $s$–$t$ path $P$ that passes through the edge $(v_s,v_t)$.
All other leaf vertices in the tree rooted at $s$ are disconnected from the tree rooted at $t$. 

An instance-smart algorithm can still ``know'' the path code corresponding to $P$ and follow it directly. However, it must additionally verify that no missing part of the tree structure could give rise to a shorter path. After inspecting all vertices in the tree rooted at $s$ at or above the level containing $v_s$, as well as all vertices in the tree rooted at $t$ at or above the level containing $v_t$, the algorithm can certify that no shorter path exists. Indeed, since all edge weights are the smallest valid positive edge weight $\ell_{\min}$, even if there were an edge connecting the two trees at the levels containing $v_s$ and $v_t$, the resulting path could not be shorter than $P$. Thus, the instance-smart algorithm requires only $O(n)$ queries.

In contrast, an algorithm that aims to be instance-optimal over all instances cannot directly identify the unique $s$–$t$ path $P$. It must instead scan the edges to locate $P$, requiring $\Omega(n\Delta)$ queries in expectation. This yields an instance-optimality ratio lower bound of $\Omega(n\Delta/n)=\Omega(\Delta)$. 
\end{proof}

%% file: paper.bib
@article{MitzenmacherV22,
  author       = {Michael Mitzenmacher and
                  Sergei Vassilvitskii},
  title        = {Algorithms with predictions},
  journal      = {Commun. {ACM}},
  volume       = {65},
  number       = {7},
  pages        = {33--35},
  year         = {2022},
  url          = {https://doi.org/10.1145/3528087},
  doi          = {10.1145/3528087},
  bibsource    = {dblp computer science bibliography, https://dblp.org}
}

@book{Dantzig1963,
author="Dantzig, George Bernard",
title="Linear Programming and Extensions",
address="Santa Monica, CA",
year="1963",
doi="10.7249/R366",
publisher="RAND Corporation"
}

@article{FredmanT87_fibonacci,
  author       = {Michael L. Fredman and
                  Robert Endre Tarjan},
  title        = {Fibonacci heaps and their uses in improved network optimization algorithms},
  journal      = {J. {ACM}},
  volume       = {34},
  number       = {3},
  pages        = {596--615},
  year         = {1987},
  url          = {https://doi.org/10.1145/28869.28874},
  doi          = {10.1145/28869.28874},
  bibsource    = {dblp computer science bibliography, https://dblp.org}
}

@article{DriscollGST88_relaxedheap,
  author       = {James R. Driscoll and
                  Harold N. Gabow and
                  Ruth Shrairman and
                  Robert Endre Tarjan},
  title        = {Relaxed Heaps: An Alternative to Fibonacci Heaps with Applications
                  to Parallel Computation},
  journal      = {Commun. {ACM}},
  volume       = {31},
  number       = {11},
  pages        = {1343--1354},
  year         = {1988},
  url          = {https://doi.org/10.1145/50087.50096},
  doi          = {10.1145/50087.50096},
  bibsource    = {dblp computer science bibliography, https://dblp.org}
}

@article{PageRank_TW,
  author       = {Mikkel Thorup and
                  Hanzhi Wang},
  title        = {Instance Optimality in PageRank Centrality Estimation},
  journal      = {CoRR},
  volume       = {abs/2512.16087},
  year         = {2025},
  url          = {https://doi.org/10.48550/arXiv.2512.16087},
  doi          = {10.48550/ARXIV.2512.16087},
  eprinttype   = {arXiv},
  eprint       = {2512.16087},
  bibsource    = {dblp computer science bibliography, https://dblp.org}
}

@phdthesis{Pohl69,
  author       = {Ira Pohl},
  title        = {Bi-directional and heuristic search in path problems},
  school       = {Stanford University, {USA}},
  year         = {1969},
  url          = {https://searchworks.stanford.edu/view/2197829},
  bibsource    = {dblp computer science bibliography, https://dblp.org}
}

@article{Dreyfus69,
  author       = {Stuart E. Dreyfus},
  title        = {An Appraisal of Some Shortest-Path Algorithms},
  journal      = {Oper. Res.},
  volume       = {17},
  number       = {3},
  pages        = {395--412},
  year         = {1969},
  url          = {https://doi.org/10.1287/opre.17.3.395},
  doi          = {10.1287/OPRE.17.3.395},
  bibsource    = {dblp computer science bibliography, https://dblp.org}
}

@article{networkscience,
    author = {Barabási, Albert-László},
    title = {Network science},
    journal = {Philosophical Transactions of the Royal Society A: Mathematical, Physical and Engineering Sciences},
    volume = {371},
    number = {1987},
    pages = {20120375},
    year = {2013},
    month = {03},
    issn = {1364-503X},
    doi = {10.1098/rsta.2012.0375},
    url = {https://doi.org/10.1098/rsta.2012.0375},
    eprint = {https://royalsocietypublishing.org/rsta/article-pdf/doi/10.1098/rsta.2012.0375/414334/rsta.2012.0375.pdf},
}

@inproceedings{BergHRRW25_esa_instanceopt_convexhull,
  author       = {Sarita de Berg and
                  Ivor {van der Hoog} and
                  Eva Rotenberg and
                  Daniel Rutschmann and
                  Sampson Wong},
  editor       = {Anne Benoit and
                  Haim Kaplan and
                  Sebastian Wild and
                  Grzegorz Herman},
  title        = {Instance-Optimal Imprecise Convex Hull},
  booktitle    = {33rd Annual European Symposium on Algorithms, {ESA} 2025, Warsaw,
                  Poland, September 15-17, 2025},
  series       = {LIPIcs},
  volume       = {351},
  pages        = {25:1--25:15},
  publisher    = {Schloss Dagstuhl - Leibniz-Zentrum f{\"{u}}r Informatik},
  year         = {2025},
  url          = {https://doi.org/10.4230/LIPIcs.ESA.2025.25},
  doi          = {10.4230/LIPICS.ESA.2025.25},
  bibsource    = {dblp computer science bibliography, https://dblp.org}
}

@article{Nicholson66,
  author       = {T. A. J. Nicholson},
  title        = {Finding the Shortest Route between Two Points in a Network},
  journal      = {Comput. J.},
  volume       = {9},
  number       = {3},
  pages        = {275--280},
  year         = {1966},
  url          = {https://doi.org/10.1093/comjnl/9.3.275},
  doi          = {10.1093/COMJNL/9.3.275},
  bibsource    = {dblp computer science bibliography, https://dblp.org}
}

@article{Dijkstra59,
  author       = {Edsger W. Dijkstra},
  title        = {A note on two problems in connexion with graphs},
  journal      = {Numerische Mathematik},
  volume       = {1},
  pages        = {269--271},
  year         = {1959},
  url          = {https://doi.org/10.1007/BF01386390},
  doi          = {10.1007/BF01386390},
  bibsource    = {dblp computer science bibliography, https://dblp.org}
}

@article{instance_optimal_Roughgarden19,
  author       = {Tim Roughgarden},
  title        = {Beyond worst-case analysis},
  journal      = {Commun. {ACM}},
  volume       = {62},
  number       = {3},
  pages        = {88--96},
  year         = {2019},
  url          = {https://doi.org/10.1145/3232535},
  doi          = {10.1145/3232535},
  bibsource    = {dblp computer science bibliography, https://dblp.org}
}

@inproceedings{FOCS_best_HaeuplerHRTT24,
  author       = {Bernhard Haeupler and
                  Richard Hlad{\'{\i}}k and
                  V{\'{a}}clav Rozhon and
                  Robert E. Tarjan and
                  Jakub Tetek},
  title        = {Universal Optimality of Dijkstra Via Beyond-Worst-Case Heaps},
  booktitle    = {65th {IEEE} Annual Symposium on Foundations of Computer Science, {FOCS}
                  2024, Chicago, IL, USA, October 27-30, 2024},
  pages        = {2099--2130},
  publisher    = {{IEEE}},
  year         = {2024},
  url          = {https://doi.org/10.1109/FOCS61266.2024.00125},
  doi          = {10.1109/FOCS61266.2024.00125},
  bibsource    = {dblp computer science bibliography, https://dblp.org}
}

@article{siamcomp_DagumKLR00,
  author       = {Paul Dagum and
                  Richard M. Karp and
                  Michael Luby and
                  Sheldon M. Ross},
  title        = {An Optimal Algorithm for Monte Carlo Estimation},
  journal      = {{SIAM} J. Comput.},
  volume       = {29},
  number       = {5},
  pages        = {1484--1496},
  year         = {2000},
  url          = {https://doi.org/10.1137/S0097539797315306},
  doi          = {10.1137/S0097539797315306},
  bibsource    = {dblp computer science bibliography, https://dblp.org}
}

@article{instnce-optimal_AfshaniBC17,
  author       = {Peyman Afshani and
                  J{\'{e}}r{\'{e}}my Barbay and
                  Timothy M. Chan},
  title        = {Instance-Optimal Geometric Algorithms},
  journal      = {J. {ACM}},
  volume       = {64},
  number       = {1},
  pages        = {3:1--3:38},
  year         = {2017},
  url          = {https://doi.org/10.1145/3046673},
  doi          = {10.1145/3046673},
  bibsource    = {dblp computer science bibliography, https://dblp.org}
}

@inproceedings{sosa_bidirectional_HaeuplerHRTT25,
  author       = {Bernhard Haeupler and
                  Richard Hlad{\'{\i}}k and
                  V{\'{a}}clav Rozhon and
                  Robert E. Tarjan and
                  Jakub Tetek},
  editor       = {Ioana Oriana Bercea and
                  Rasmus Pagh},
  title        = {Bidirectional Dijkstra's Algorithm is Instance-Optimal},
  booktitle    = {2025 Symposium on Simplicity in Algorithms, {SOSA} 2025, New Orleans,
                  LA, USA, January 13-15, 2025},
  pages        = {202--215},
  publisher    = {{SIAM}},
  year         = {2025},
  url          = {https://doi.org/10.1137/1.9781611978315.16},
  doi          = {10.1137/1.9781611978315.16},
  bibsource    = {dblp computer science bibliography, https://dblp.org}
}

@inproceedings{stoc_ValiantV16_instance_optimal,
  author       = {Gregory Valiant and
                  Paul Valiant},
  editor       = {Daniel Wichs and
                  Yishay Mansour},
  title        = {Instance optimal learning of discrete distributions},
  booktitle    = {Proceedings of the 48th Annual {ACM} {SIGACT} Symposium on Theory
                  of Computing, {STOC} 2016, Cambridge, MA, USA, June 18-21, 2016},
  pages        = {142--155},
  publisher    = {{ACM}},
  year         = {2016},
  url          = {https://doi.org/10.1145/2897518.2897641},
  doi          = {10.1145/2897518.2897641},
  bibsource    = {dblp computer science bibliography, https://dblp.org}
}

@inproceedings{focs_NarayananRTT24_thorup,
  author       = {Shyam Narayanan and
                  V{\'{a}}clav Rozhon and
                  Jakub Tetek and
                  Mikkel Thorup},
  title        = {Instance-Optimality in I/O-Efficient Sampling and Sequential Estimation},
  booktitle    = {65th {IEEE} Annual Symposium on Foundations of Computer Science, {FOCS}
                  2024, Chicago, IL, USA, October 27-30, 2024},
  pages        = {658--688},
  publisher    = {{IEEE}},
  year         = {2024},
  url          = {https://doi.org/10.1109/FOCS61266.2024.00048},
  doi          = {10.1109/FOCS61266.2024.00048},
  bibsource    = {dblp computer science bibliography, https://dblp.org}
}

@inproceedings{focs_BlancCW25_clement,
  author       = {Guy Blanc and
                  Cl{\'{e}}ment L. Canonne and
                  Erik Waingarten},
  title        = {Instance-Optimal Uniformity Testing and Tracking},
  booktitle    = {66th {IEEE} Annual Symposium on Foundations of Computer Science, {FOCS}
                  2025, Sydney, Australia, December 14-17, 2025},
  pages        = {1351--1365},
  publisher    = {{IEEE}},
  year         = {2025},
  url          = {https://doi.org/10.1109/FOCS63196.2025.00071},
  doi          = {10.1109/FOCS63196.2025.00071},
  bibsource    = {dblp computer science bibliography, https://dblp.org}
}

@article{siamcomp/ValiantV17,
  author       = {Gregory Valiant and
                  Paul Valiant},
  title        = {An Automatic Inequality Prover and Instance Optimal Identity Testing},
  journal      = {{SIAM} J. Comput.},
  volume       = {46},
  number       = {1},
  pages        = {429--455},
  year         = {2017},
  url          = {https://doi.org/10.1137/151002526},
  doi          = {10.1137/151002526},
  bibsource    = {dblp computer science bibliography, https://dblp.org}
}

@incollection{books_ValiantV20,
  author       = {Gregory Valiant and
                  Paul Valiant},
  editor       = {Tim Roughgarden},
  title        = {Instance Optimal Distribution Testing and Learning},
  booktitle    = {Beyond the Worst-Case Analysis of Algorithms},
  pages        = {506--526},
  publisher    = {Cambridge University Press},
  year         = {2020},
  url          = {https://doi.org/10.1017/9781108637435.029},
  doi          = {10.1017/9781108637435.029},
  bibsource    = {dblp computer science bibliography, https://dblp.org}
}

@inproceedings{colt_ChenL16_openproblem,
  author       = {Lijie Chen and
                  Jian Li},
  editor       = {Vitaly Feldman and
                  Alexander Rakhlin and
                  Ohad Shamir},
  title        = {Open Problem: Best Arm Identification: Almost Instance-Wise Optimality
                  and the Gap Entropy Conjecture},
  booktitle    = {Proceedings of the 29th Conference on Learning Theory, {COLT} 2016,
                  New York, USA, June 23-26, 2016},
  series       = {{JMLR} Workshop and Conference Proceedings},
  pages        = {1643--1646},
  publisher    = {JMLR.org},
  year         = {2016},
  url          = {http://proceedings.mlr.press/v49/chen16b.html},
  bibsource    = {dblp computer science bibliography, https://dblp.org}
}

@inproceedings{nips_LiRNJJ22,
  author       = {Zhaoqi Li and
                  Lillian J. Ratliff and
                  Houssam Nassif and
                  Kevin Jamieson and
                  Lalit Jain},
  editor       = {Sanmi Koyejo and
                  S. Mohamed and
                  A. Agarwal and
                  Danielle Belgrave and
                  K. Cho and
                  A. Oh},
  title        = {Instance-optimal {PAC} Algorithms for Contextual Bandits},
  booktitle    = {Advances in Neural Information Processing Systems 35: Annual Conference
                  on Neural Information Processing Systems 2022, NeurIPS 2022, New Orleans,
                  LA, USA, November 28 - December 9, 2022},
  year         = {2022},
  url          = {http://papers.nips.cc/paper\_files/paper/2022/hash/f4821075019a058700f6e6738eea1365-Abstract-Conference.html},
  bibsource    = {dblp computer science bibliography, https://dblp.org}
}

@inproceedings{dang_neurips_23,
author = {Dang, Trung and Lee, Jasper C.H. and Song, Maoyuan and Valiant, Paul},
title = {Optimality in mean estimation: beyond worst-case, beyond sub-Gaussian, and beyond 1 + $\alpha$ moments},
year = {2023},
publisher = {Curran Associates Inc.},
address = {Red Hook, NY, USA},
booktitle = {Proceedings of the 37th International Conference on Neural Information Processing Systems},
articleno = {182},
numpages = {27},
location = {New Orleans, LA, USA},
series = {NIPS '23}
}

@article{FaginLN03_instance_optimal,
  author       = {Ronald Fagin and
                  Amnon Lotem and
                  Moni Naor},
  title        = {Optimal aggregation algorithms for middleware},
  journal      = {J. Comput. Syst. Sci.},
  volume       = {66},
  number       = {4},
  pages        = {614--656},
  year         = {2003},
  url          = {https://doi.org/10.1016/S0022-0000(03)00026-6},
  doi          = {10.1016/S0022-0000(03)00026-6},
  bibsource    = {dblp computer science bibliography, https://dblp.org}
}

@article{hart1968formal,
  title={A formal basis for the heuristic determination of minimum cost paths},
  author={Hart, Peter E and Nilsson, Nils J and Raphael, Bertram},
  journal={IEEE transactions on Systems Science and Cybernetics},
  volume={4},
  number={2},
  pages={100--107},
  year={1968},
  publisher={IEEE}
}
